\documentclass[9pt,journal,twoside,web]{ieeecolor}

\usepackage{generic}
\usepackage{cite}
\usepackage{amsmath,amssymb,amsfonts,amsthm}
\usepackage{graphicx}
\usepackage{algorithm,algorithmicx}
\usepackage{hyperref}
\usepackage{booktabs}
\hypersetup{hidelinks=true}
\usepackage{textcomp}
\usepackage{algpseudocode}
\usepackage{cuted}
\usepackage{subfigure}
\usepackage{multirow}
\usepackage{stfloats}

\newtheorem{theorem}{Theorem}[section]
\newtheorem{lemma}{Lemma}[section]
\newtheorem{proposition}{Proposition}[section]
\newtheorem{remark}{Remark}[section]
\newtheorem{define}{Definition}[section]

\makeatletter
\newcommand\approxsim{\mathchoice
  {\@approxsim {\displaystyle}      {1ex} }
  {\@approxsim {\textstyle}         {1ex} }
  {\@approxsim {\scriptstyle}       {.7ex}}
  {\@approxsim {\scriptscriptstyle} {.5ex}}}
\newcommand\@approxsim[2]{%
  \mathrel{%
    \ooalign{%
      $\m@th#1\sim$\cr
      \hidewidth$\m@th#1.$\hidewidth\cr
      \hidewidth\raise #2 \hbox{$\m@th#1.$}\hidewidth\cr
    }%
  }%
}
\makeatother

\def\BibTeX{{\rm B\kern-.05em{\sc i\kern-.025em b}\kern-.08em
    T\kern-.1667em\lower.7ex\hbox{E}\kern-.125emX}}
\begin{document}
\title{On Fast Attitude Filtering Using Matrix Fisher Distributions with Stability Guarantee as an Observer}
\author{Shijie Wang, Haichao Gui, and Rui Zhong
\thanks{This work was supported in part by the National Natural Science
Foundation of China under Grant 12422214.}
\thanks{Shijie Wang is with School of Astronautics, Shenyuan Honors College, Beihang University, Beijing, 100191, China (e-mail: wangsj2001@buaa.edu.cn). }
\thanks{Haichao Gui is with School of Astronautics and the MOE Key Laboratory of Spacecraft Design, Optimization and Dynamic Simulation, Beihang University, Beijing, 100191, China (e-mail: hcgui@buaa.edu.cn, corresponding author). }
\thanks{Rui Zhong is with School of Astronautics, Beihang University, Beijing, 100191, China (e-mail: zhongruia@163.com). }}

\maketitle
\begin{abstract}
This paper addresses two interrelated problems of the nonlinear filtering mechanism and fast attitude filtering with the matrix Fisher distribution (MFD) on the special orthogonal group. By analyzing the distribution evolution along Bayes' rule, we reveal two essential properties that enhance the performance of Bayesian attitude filters with MFDs, particularly in challenging conditions. Benefiting from the new understanding of the filtering mechanism associated with MFDs, two closed-form filters with MFDs are then proposed. These filters avoid the burdensome computations in previous MFD-based filters by introducing linearized error systems with right-invariant errors but retaining the two advantageous properties. The proposed filter with right-invariant error is shown to be almost globally asymptotically stable for any trajectory on $SO(3)$ as a deterministic observer, leveraging its Wahba-type update and right-invariant construction. Moreover, we further prove the local exponential stability of the filter for single-axis rotations to reveal the effect of the two properties on the convergence rate.  These stability results support the performance of the proposed filter with large initial error from a theoretical viewpoint, which to our knowledge, is not achieved
by existing directional statistics-based filters.  Numerical simulations demonstrate that proposed filters are as accurate as recent MFD-based Bayesian filters in challenging circumstances but consume far less computation time (by up to about two orders of magnitude compared with previous MFD-based attitude filters).
\end{abstract}


\section{Introduction}
\label{sec:introduction}
 \IEEEPARstart{T}{he} problem of attitude estimation with uncertain noisy measurements is a core task in the navigation of mobile robotics, aeronautics, astronautics, and other realms. In these realms, many systems, such as unmanned aerial vehicles and nanosatellites, have limited computational resources and low-cost sensors. Meanwhile, the unique fact that attitude evolves on the special orthogonal group $SO(3)$, makes it difficult for attitude estimators to accurately handle large errors and uncertainties with low computational complexity. As a result, it is valuable and challenging to design efficient attitude estimators on $SO(3)$ with limits of computational resources.

Attitude observers, as deterministic approaches, are designed directly on $SO(3)$, avoiding singularities of local coordinates and achieving theoretical guarantees of stability and convergence speed. In \cite{mahony2008nonlinear}, non-linear complementary filters were proposed and proved to be almost globally stable with the Lyapunov analysis. Higher-order nonlinear complementary filters were proposed in \cite{zlotnik2018higher} for estimation on Lie group. More recently, hybrid observers were designed for attitude estimation, providing stronger forms of stability, such as global asymptotic stability \cite{wu2015globally} and global exponential stability \cite{berkane2017hybrid}. The hybrid observers were further developed for even more formulations, e.g., simultaneous attitude and gyro bias estimation with global exponential stability \cite{berkane2017hybridbias}, pose estimation using landmark measurements \cite{wang2020hybrid}, attitude estimation using intermittent and multirate vector measurements \cite{wang2023nonlinear}, etc. A notable advantage of these approaches is that they admit rigorous Lyapunov-based proofs of stability and convergence rates, which makes them highly reliable and theoretically sound. Instead of deterministic observers, this paper focuses on stochastic estimation techniques. Compared to deterministic observers, they take stochastic properties into account and can describe the uncertainty of estimated states in terms of probability density.

For attitude estimation, the classic multiplicative extended Kalman filter (MEKF) was developed in \cite{leffertsKalmanFilteringSpacecraft1982} and \cite{markley2003attitude}, representing attitude by quaternion and uncertainty distribution around the mean attitude by a Gaussian distribution on the three-dimensional Euclidean space. More recently, the invariant extended Kalman filter (IEKF) was developed in \cite{barrauIntrinsicFilteringLie2015,barrauInvariantExtendedKalman2017,barrauStochasticObserversLie2018,barrau2016invariant,barrauExtendedKalmanFiltering2020,barrau2022geometry} by utilizing more geometric and algebraic properties of Lie groups and was applied to estimate attitude for spacecraft navigation, 
 \cite{guiQuaternionInvariantExtended2018}, SLAM \cite{barrau2015ekf,liuInGVIOConsistentInvariant2023,songRightInvariantExtended2022} and other industrial applications. In the IEKF, the Gaussian distributions in exponential coordinates (a coordinate system with 
the exponential map near the group identity \cite{chirikjian2011stochastic}) are used to represent the uncertainty of states \cite{barrauInvariantKalmanFiltering2018} and the linearized error system is adopted to propagate uncertainty. Another Gaussian distribution in exponential coordinates, referred to as the concentrated Gaussian distribution (CGD) \cite{Wang2006ErrorPO, wang2008nonparametric, wolfe2011bayesian} is also used to represent 
the uncertainty of states and to construct Kalman filters (KF) on Lie groups, such as CD-LR-EKF  \cite{Bourmaud2014ContinuousDiscreteEK}. 
When applying these two distributions to derive variants of the EKF, the resultant filtering gain tuning and covariance update still depend on the linearized error system and the properties of Gaussian distributions on the n-dimensional Euclidean space, which remain the same as the KF. This implies that the state error is still treated as a random variable following Gaussian distributions. This strategy is effective for systems with sufficient \textit{a priori} information and small measurement errors but limits the performance of the IEKF and other similar filters in challenging circumstances with large initial errors and measurement uncertainties. For example, it is difficult for the Gaussian distribution in local coordinates to characterize the attitude with high uncertainty due to wrapped errors
 \cite{wang2021matrix}.

To address the shortcomings brought by the structure of filters with Gaussian distributions in the Euclidean space, the robotics community has increasingly recognized the importance of properly defining distributions on Lie groups \cite{barrauInvariantKalmanFiltering2018}. A series of notable approaches to construct filters with distributions on $SO(3)$ are assumed density filters with directional statistics, a branch of statistics dealing with observations lying on the nonlinear manifold \cite{mardiaDirectionalStatistics1999,leyModernDirectionalStatistics2017,leyAppliedDirectionalStatistics2019}. The assumed density filters for attitude estimation make use of the structure of Bayesian filters and the distributions provided by directional statistics to approximate the underlying posterior distribution of attitude, such as the matrix Fisher distribution (MFD) on Stiefel manifolds \cite{downsOrientationStatistics1972,khatriMisesFisherMatrixDistribution1977} and the Bingham distribution on the n-dimensional sphere \cite{binghamAntipodallySymmetricDistribution1974}. In \cite{leeBayesianAttitudeEstimation2018}, Bayesian attitude filters were designed with MFDs on $SO(3)$, using the first moment to propagate distributions along the kinematics. Meanwhile, several attitude filters have been constructed based on Bingham distributions on the three-dimensional sphere \cite{kurzUnscentedMisesFisher2016,wangBinghamGaussianDistributionBackslash2021}. Recently, attitude filters were proposed with a newly defined matrix Fisher-Gaussian distribution on $SO(3) \times \mathbb{R}^n$ in \cite{wang2020matrix} and \cite{wang2021matrix}, with methodologies similar to those in \cite{leeBayesianAttitudeEstimation2018}.

A range of assumed density filters with the MFDs demonstrate superior performance compared to the classic MEKF in challenging circumstances but have computational burdens that far exceed the MEKF. Although approximate solutions were presented in \cite{leeBayesianAttitudeEstimation2018b} to accelerate computations, these filters still spend much more computation time than variants of the EKF, as shown in Section \ref{sec: Simulation}. Furthermore, the reason why filters with MFDs perform better than variations of the EKF such as the MEKF in challenging circumstances, still needs more in-depth investigations. In \cite{wang2021matrix}, it was intuitively attributed to the reasons that these works do not use linearization to construct estimators and avoid the wrapped errors caused by Gaussian distributions on local coordinates. There is, however, a lack of comprehensive theoretical explanations on how wrapped errors affect attitude estimation and why the uncertainties in attitude estimation are increased by the MFD-based estimator when initial errors are large
 and the initial confidence is falsely high, as observed in \cite{leeBayesianAttitudeEstimation2018,wang2020matrix}. What mechanisms make the MFD-based filters perform better and how to design fast MFD-based filters are two closely interrelated problems. Once the crucial properties associated with the advantageous performance of MFD-based filters are identified mathematically, new filters can be designed by preserving these properties and streamlining the rest for faster computations. 

In this paper, we reveal two properties intrinsic in MFDs that improve the performance of MFD-based attitude filters, and then propose two fast MFD-based filters. The proposed filters retain the two properties while introducing linearized error systems to significantly enhance computational speed. As a result, they are almost as accurate as previous MFD-based attitude filters but consume far less computational time. The proposed filter with right-invariant error is endowed with almost global asymptotic stability for any trajectory on $SO(3)$ as a deterministic observer with noise-free measurements, and local exponential stability for the case of single-axis rotations. To the best of our knowledge, there are no results in the literature about estimators with directional statistics achieving such stability properties. The detailed methodologies and contributions are summarized as follows.

In Section \ref{sec:filter_mec}, the evolution of the distribution on $SO(3)$ along Bayes' rule is studied, which indicates how the statistical characteristics of the prior and the likelihood distributions impact the posterior distribution. To further characterize the distribution on $SO(3)$, we define two new parameters for both MFDs and CGDs, which are referred to as the mean angle and concentration parameter respectively. The nonlinear recurrence relation of the two parameters describe how the differences in mean attitudes affect the posterior distribution.  By comparing their expressions for the posterior distribution derived from the Bayesian filter with MFDs and CGDs, two key properties of the MFD-based Bayesian filter are identified mathematically:
\begin{itemize}
    \item  Due to the special correlation of the two parameters, the difference in the mean attitudes of the prior and the likelihood distributions raises the uncertainty of the posterior distribution for rotations around a specific axis.
    \item  Because of the nonlinear recurrence relations of the parameters, the mean angle of the posterior distribution from the filter with MFDs is always closer to a higher-confidence mean angle (between the likelihood and the prior distributions) than that from the filter with CGDs.

\end{itemize}
The above properties endows MFD-based filters with robustness to initial errors. The two properties are further connected with the exponential convergence rate for single-axis rotation in Section \ref{sec:stability}.


These two properties are intrinsic to MFDs and independent of the accuracy of the method used to approximate the underlying prior and likelihood distributions (such as the first-moment matching or the unscented transformation). Thus, in Section \ref{sec:Filter}, we construct two fast nonlinear filters with MFDs on $SO(3)$, which utilize linearization to approximate the prior and the likelihood distributions but retain the two key properties. We establish two probability models with left- and right-invariant errors respectively to represent the uncertainty of attitude, and both two invariant errors are characterized by MFDs, instead of Gaussian distributions in exponential coordinates as adopted by the IEKF. In the filtering procedure, the parameters of MFDs are calculated through linearized error equations to approximate the underlying prior and likelihood distributions with first-order precision. 
Then the prior and the likelihood distributions, approximated by MFDs, are fused with Bayes' rule to obtain the posterior distribution on $SO(3)$. Due to the new understanding of the properties of Bayesian filters with MFDs, the two filters retain the key factors enabling higher estimation accuracy while avoiding solving nonlinear equations, as required by previous MFD-based estimators \cite{leeBayesianAttitudeEstimation2018}, for uncertainty propagation, which significantly reduces the computational complexity of the proposed filters. Then, the stability properties of the proposed filter with right-invariant error as a deterministic observer are studied in Section \ref{sec:stability}. The observer form of the proposed filter is proven to be almost globally asymptotically stable on $SO(3)$, presenting the first stability results for estimators based on directional statistics. Additionally, the one-dimensional form of the proposed filter is derived for a single-axis rotation case, leading to an upper bound on the exponential convergence rate, which establishes the direct relation between the two properties and the convergence rate.

In challenging circumstances with large initial error and measurement uncertainty, numerical simulations presented in Section \ref{sec: Simulation} illustrate that the two proposed filters exhibit almost identical accuracy as the estimator in \cite{leeBayesianAttitudeEstimation2018}. Moreover, the proposed filters consume about one to two orders of magnitude less computation time than the filter in \cite{leeBayesianAttitudeEstimation2018}, and about one order of magnitude less than that in \cite{leeBayesianAttitudeEstimation2018b}. The conclusion is drawn in Section \ref{sec: Conclusion}.

\section{Preliminaries}
\label{sec:Preliminaries}
\subsection{Notations and Some Facts}
Throughout this paper,  $\mathbb{R}^n$ denotes the $n$-dimensional Euclidean space, and $I_n$ denotes the $n \times n$ identity matrix. Consider a rigid body with its body-fixed frame $\mathcal{F}_B = \{b_1, b_2, b_3\}$, where the vector $b_i \in \mathbb{R}^3$ is the $i$-th basis of $\mathcal{F_B}$. Similarly, we define an inertial frame and its basis $\mathcal{F}_I = \{e_1, e_2, e_3\}$.  Then the attitude of the rigid body is the orientation of $\mathcal{F_B}$ relative to $\mathcal{F_I}$, which can be represented by a rotation matrix $R\in \mathbb{R}^{3\times 3}$ given by $R_{ij}=e_i\cdot b_j$, where $R_{ij}$ is the element of the $i$-th row and $j$-th column of $R$. When both $\mathcal{F_B}$ and $\mathcal{F_I}$ are right-handed orthonormal frames, the rotation matrices with matrix multiplication form the special orthogonal group
\[
SO(3) = \{ R\in \mathbb{R}^{3 \times 3} | R^TR=I_3, \text{det}[R]=1 \},
\]
where $\text{det}[\cdot]$ gives the determinant of a matrix. The tangent space at $R \in SO(3)$ is denoted by $T_RS O(3)$. The Lie algebra of $SO(3)$ is denoted by $\mathfrak{so}(3) = \left\{ { X \in \mathbb{R} ^{3\times 3} | X = -X^T } \right\}$,  and the map: $\wedge : \mathbb{R}^3 \rightarrow \mathfrak{so}(3)$ is an isometry between $\mathfrak{so}(3)$ and $\mathbb{R} ^{3}$. The hat map is defined such that $x^{\wedge}\ y = x\times y$ for any $x,y\in \mathbb{R}^3$, and its inverse map is denoted by $\vee$. For $n\times n$ symmetric matrices $A$ and $B$, $A\preceq B$ denotes that $B-A$ is positive semi-definite, and the notation $A \prec B$ means $B-A$ is positive definite.

The notations $\tilde{a}$ and $\hat{a}$ denote a measured variable and an estimated variable respectively. The subscript $k$ denotes the $k$-th time-step of the time sequence $\{ t_0, t_1, \dots \}$ and the superscript $i$ denotes the $i$-th term of the set $\{ a^1, a^2, \dots \}$. The trace of a matrix is denoted by $\text{tr}(\cdot)$, and the function $e^{\text{tr}(\cdot)}$ is abbreviated as $\text{etr}(\cdot)$. The sign function is denoted by ${\rm sign(\cdot)}$. The Frobenius norm of a matrix $M\in \mathbb{R}^{3\times 3}$ is defined as $\Vert M \Vert_F = \sqrt{\text{tr}\left( M^TM \right)}$. 


The following equations involving the map $\wedge$ and trace will be repeatedly used. For any $A\in \mathbb{R}^{3\times3}$, $R\in SO(3)$, and $a,b \in \mathbb{R}^{3}$, we have 
\begin{align}
    \left(Ra\right)^{\wedge} &= Ra^{\wedge}R^T \label{eq:crossTr1} \\
    \left( a^{\wedge}A+A^Ta^{\wedge} \right) &= [\left({\rm tr}\left(A\right)I_3-A\right)a]^{\wedge}, \label{eq:crossTr2} \\
    \left(ba^T- ab^T\right)^{\vee} &= a^{\wedge}b. \label{eq:crossTr3}
\end{align}
\subsection{Distributions on SO(3)}

\subsubsection{The matrix Fisher distribution}
The MFD is a compact exponential density on the Stiefel manifold. In recent years, the MFD has been applied to $SO(3)$ to represent the uncertainty of attitude, which was introduced initially in \cite{leeBayesianAttitudeEstimation2018}. The probability density function (PDF) of a random rotation matrix $R \sim \mathcal{M}(F)$ is defined as
\begin{equation}
    p(R\in SO(3);F) = \frac{1}
{{c\left( F \right)}}\text{etr}\left[ {F^T R} \right],\label{eq:de_MF}
\end{equation}
where $F \in \mathbb{R}^{3\times3}$ is a parameter and $c(F)\in \mathbb{R}$ is the normalized constant. $c(F)$ is calculated by $c(F) = \int_{SO(3)} {\text{etr}\left[{F^TR}\right]dR}$, where $dR$ is the bi-invariant Haar measure. The parameter $F$ can be decomposed via the proper singular value decomposition (SVD) \cite{markleyAttitudeDeterminationUsing1988} as follows
\begin{equation}
    F=USV^T,\label{eq:usv}
\end{equation}
where $U,V \in SO(3)$ and $S = \text{diag}[s_1, s_2, s_3]\in \mathbb{R}^{3\times3}$ with $s_1 \geq s_2 \geq |s_3| \geq 0$. Given a random rotation matrix $R \sim \mathcal{M}(F)$, the \textit{mean attitude} of $R$ is defined as the attitude that maximizes the density function and minimizes the Frobenius mean squared error, given by $M = UV^T \in SO(3)$.

The $i$-th principal axis of the MFD is studied in \cite{leeBayesianAttitudeEstimation2018}, which is $Ue_i$ when resolved in $\mathcal{F}_I$, or $Ve_i$ when resolved in $\mathcal{F}_B$. The parameter $s_j + s_k$ illustrates the concentration of dispersion relative to the $i$-th principal axis. We can also apply the proper right polar decomposition to $F$, which gives
\begin{equation}
    F = MK,\label{eq:polar_r}
\end{equation}
where $M\in SO(3)$ is called the polar component  and $K\in \mathbb{R}^{3 \times 3}$ is a symmetric positive-definite matrix referred to as the elliptical component. Compared with (\ref{eq:usv}),  it follows that $M=UV^T$ and $K=VSV^T$, which represent the center of the distribution, and the dispersion relative to the mean attitude, respectively. Similarly, the proper left polar decomposition of $F$ is given by
\begin{equation}\label{eq:polar_l}
    F=K^{\prime} M^{\prime},
\end{equation}
where $K^{\prime} = USU^T$ and $M^{\prime}=UV^T$.

Next, we present the transformation of MFDs under rotation and transpose.
\begin{lemma}\label{le:rotation}
    Given fixed matrices $R_l,R_r \in SO(3)$, and $R \sim \mathcal{M}(F)$, the random rotation matrix $R^{\prime} = R_l R R_r \in SO(3)$ follows a MFD on $SO(3)$ with
    \begin{equation}
        R^{\prime} \sim \mathcal{M}\left( R_l F R_r \right).\label{eq:trans}
    \end{equation}
\end{lemma}
\begin{proof}
    Since $dR$ is the bi-invariant Haar measure, we can derive
    \begin{align*}
        p(R^{\prime}) &\propto {\rm etr}\left( F^T R \right)={\rm etr}\left( F^T R_l^T R^{\prime} R_r^T \right)\\
             &=       {\rm etr}\left((R_l F R_r)^T R^{\prime}  \right),
    \end{align*}
    which shows (\ref{eq:trans}).
\end{proof}

\begin{proposition}
    Suppose that $R \sim \mathcal{M}(N)$ and $N \in \mathbb{R}^{3 \times 3}$ is a symmetric matrix. The random rotation matrix $R_T=R^T$ is characterized by the following probability,
    \begin{equation}\label{eq:trans2}
        p(R_T) = \frac{1}{c(N)} {\rm etr}\left( N R_T \right),
    \end{equation}
    i.e., $R_T \sim \mathcal{M}(N)$ as well.
\end{proposition}
\begin{proof}
    Using the Jacobian matrix $\partial R / \partial R_T = I_3$, the density of $R_T$ is
    \begin{equation*}
        p(R_T) \propto {\rm etr}\left( N^T R \right) = {\rm etr}\left( R_T N \right) = {\rm etr}\left( N R_T \right),
    \end{equation*}
    which shows (\ref{eq:trans2}).
\end{proof}


\subsubsection{The concentrated Gaussian distribution}
The CGD \cite{Wang2006ErrorPO, wang2008nonparametric, wolfe2011bayesian,Bourmaud2014ContinuousDiscreteEK} is a generalization of the Gaussian distribution on Lie groups. Following \cite{wolfe2011bayesian}, the PDF of a random rotation matrix $R$ following CGD is defined as
\begin{subequations}
    \begin{align}\label{eq:left_CGD}
        &p(R\in SO(3);M,P) \nonumber\\
        &= \frac{1}{c} \exp \left\{-\frac{1}{2}  \left[ \log \left(M^TR\right) ^{\vee}\right]^T (P)^{-1} \left[ \log \left(M^TR\right) ^{\vee}\right] \right\},
    \end{align}
or
\begin{align}\label{eq:right_CGD}
        &p(R\in SO(3);M,P) \nonumber\\
        &= \frac{1}{c} \exp \left\{-\frac{1}{2}  \left[ \log \left(RM^T\right) ^{\vee}\right]^T (P)^{-1} \left[ \log \left(RM^T\right) ^{\vee}\right] \right\},
    \end{align}
\end{subequations}
where $c$ is the normalized constant, $M \in SO(3)$ is the mean attitude and $P \in \mathbb{R}^{3 \times 3}$ is a symmetric matrix. The rotation matrices following (\ref{eq:left_CGD}) or (\ref{eq:right_CGD}) are denoted by $R\sim\mathcal{N}_L(M,P)$ or $R\sim\mathcal{N}_R(M,P)$. 

\subsection{System Models}
Let $R \in SO(3)$ represent the attitude of a rigid body with respect to $\mathcal{F}_I$. The continuous attitude kinematics without noise in terms of rotation matrix is written as
\begin{equation}
    \label{eq:AttitudeKm}
    R_t^T dR_t =\omega_tdt,
\end{equation}
where $\omega_t \in \mathbb{R}^3$ is the angular velocity relative to  $\mathcal{F}_I$ expressed in $\mathcal{F}_B$. In practical attitude estimation tasks, the angular velocity is often measured by rate gyroscopes, and its typical measurement model without bias is given by
\begin{equation}
\label{eq:gyro}
    \tilde{\omega}_t = \omega_t - HdW,
\end{equation}
where $dW$ is a 3-dimensional Wiener process and $H \in \mathbb{R}^{3 \times 3}$ describes the magnitude of noise. The kinematics (\ref{eq:AttitudeKm}) and (\ref{eq:gyro}) can be discretized with time sequence $\{ t_0, t_1, \dots, t_k, \dots \}$ and the fixed time step $h=t_k-t_{k-1}$ as
\begin{equation}\label{eq:AttKiDe}
    R_{k} = R_{k-1} \exp\left\{ (h\tilde{\omega}_k + w_k)^{\wedge} \right\},
\end{equation}
where $w_k = H_k\Delta W$ is a centered Gaussian with a covariance matrix $Q_k = hH_kH_k^T$. Additionally, vector measurements  are usually used to correct the accumulation error of gyros, and the corresponding measurement model is given by
\begin{equation}\label{eq:AttMea}
    \tilde{b}^i_k = R_k^T e_k^i + v^i_k, i=1,2,\dots,n
\end{equation}
where $e_k^i \in \mathbb{R}^3$ is a known vector expressed in $\mathcal{F}_I$, $\tilde{b}_k^i$ is the corresponding measurement expressed in $\mathcal{F}_B$ and $v^i_k$ is a centered Gaussian noise with a covariance matrix $G_k^i$.

\section{Nonlinear Filtering Mechanism Caused by MFDs }
\label{sec:filter_mec}
In this section, we analyze the evolution of the distribution on $SO(3)$ along Bayes' rule, which leads to the differences in the filtering mechanism between MFD-based Bayesian filters and CGD-based Bayesian filters. 

\subsection{Distributions on the subset of SO(3)}\label{sec:para}

We focus on a special subset containing the mean attitudes of both the prior and the likelihood distributions, rather than the entire $SO(3)$, such that how the difference in mean attitudes of the prior and the likelihood distributions influences the posterior distribution are highlighted. The following defines a subset generated by the mean attitude rotating around a specific axis.

\begin{define}\label{de:subset}
  Consider a deterministic attitude represented by $M$. The subset of $SO(3)$ generated by $M$ rotating around a specific axis $w_0$ is denoted by $\mathcal{S}_M (w_0)$, which is defined as
    \begin{equation}\label{eq:def_subset}
       \mathcal{S}_M (w_0)\! \triangleq\! \{ R(\theta)\!\in\! SO(3) | R\!(\theta)\! = \!M \exp (\theta \omega_0 ^{\wedge}),\!\ \theta \!\in\! [ -\pi,\pi )\}.
    \end{equation}
\end{define}

\begin{proposition}\label{the:subset_2}
    Consider any two rotation matrices $M_1,\ M_2 \in SO(3)$. There always exists a subset $\mathcal{S}_{M_0 }(w_0)$ such that $M_1,\ M_2 \in \mathcal{S}_{M_0 }(w_0)$, where $M_0$ and $w_0$ are given by
    \begin{equation*}
        M_0=M_1,\ w_0=\frac{\log\left(M_1^TM_2\right)^{\vee}}{\Vert 
\log\left(M_1^TM_2\right)^{\vee} \Vert}
    \end{equation*}
\end{proposition}
\begin{proof}
    Substituting $\theta = \Vert 
\log\left(M_1^TM_2\right)^{\vee} \Vert$ into (\ref{eq:def_subset}), we have
\begin{equation*}
    R(\theta) = M_1 \exp\left( \log\left(M_1^TM_2\right) \right) = M_2,
\end{equation*}
which indicates $M_2 \in \mathcal{S}_{M_0 }(w_0)$.
\end{proof}
Definition \ref{de:subset} and Proposition \ref{the:subset_2} can be viewed as an equivalence of Euler's rotation theorem. The proposition implies that the mean attitudes of any two distributions can be included in the same subset. Besides, the difference between the mean attitudes of two distributions, such as the prior and the likelihood distributions, can be represented by a one-dimensional scalar $\theta$.

Then, the probability density functions of both MFDs and CGDs on the subset defined by Definition \ref{de:subset} are each proven to be single variable functions of $\theta$ with two parameters.

\begin{theorem}\label{the:subset_1}
    Consider a random attitude $R\sim \mathcal{M}(F)$ and the proper right polar decomposition of $F$ is given by $F=MK$. Let $\mathcal{S}_{M_0 }(w_0)$ be a subset of $SO(3)$ defined as (\ref{eq:def_subset}). If  the mean attitude $M$ belongs to $\mathcal{S}_{M_0 }(w_0)$, the probability density of  the elements in $\mathcal{S}_{M_0 }(w_0)$ can be rewritten as
    \begin{align}\label{eq:MF2vMF}
        p(R(\theta)) \propto \exp \left[ \left( {\rm tr}(K) - w_0^TKw_0 \right)\cos(\theta-\bar{\theta}) \right],
    \end{align}
    where $\bar{\theta} ={w_0 \cdot  \log(M_0M^T)^{\vee}}$.
\end{theorem}
\begin{proof}
   Because $M\in\mathcal{S}_{M_0 }(w_0)$, we have $M = M_0 \exp (\bar{\theta} \omega_0 ^{\wedge})$.
    Then every $R$ belonging to $\mathcal{S}_{M_0 }(w_0)$ can be rewritten as
    \begin{align}\label{eq_M2M_0}
        R(\theta) = M \exp \left((\theta-\bar{\theta}) \omega_0 ^{\wedge}\right).
    \end{align}
    Substituting (\ref{eq_M2M_0}) into (\ref{eq:de_MF}) and using Rodrigues' rotation formula \cite{chirikjianEngineeringApplicationsNoncommutative2000}, we have
        \begin{align}\label{eq:derive_vmf}
  p(R(\theta)) &= \frac{1}
{{c(F)}}{\text{etr}}\left[ {KM^T M\exp \left( {(\theta  - \bar \theta )\omega _0^ \wedge  } \right)} \right] \nonumber \\
   &= \frac{1}
{{c(F)}}{\text{etr}}\left[ {K\cos(\theta  - \bar \theta ) + Kw_0^ \wedge  \sin(\theta  - \bar \theta )} \right. \nonumber \\
  & + Kw_0 w_0^T \left( {1 - \cos(\theta  - \bar \theta )} \right) \nonumber \\
   &= \frac{1}{{c(F)}}{\rm{exp}}\left[ {{\rm{tr}}(K)\cos(\theta  - \bar \theta ) + 0 + } \right. \nonumber \\
   & \left. {w_0^TK{w_0}\left( {1 - \cos(\theta  - \bar \theta )} \right)} \right] \nonumber\\
   &= \frac{{e^{w_0^T Kw_0 } }}
{{c(F)}}{\text{exp}}\left[ {\left( {{\text{tr}}(K) - w_0^T Kw_0 } \right)\cos(\theta  - \bar \theta )} \right], 
    \end{align}
    which gives (\ref{eq:MF2vMF}).
\end{proof}
The PDF on $\mathcal{S}_{M_0 }(w_0)$ given by (\ref{eq:MF2vMF}) is similar to the von-Mises distribution on $\mathbb{S}^2$ \cite{fisherDispersionSphere1953}, which is defined as
\begin{equation}\label{eq:vmf}
    p(\theta) \propto e^{\kappa \cos(\theta - \bar{\theta})},
\end{equation}
where the angle parameter $\bar{\theta}$ maximizes the probability $p(\theta)$  and $\kappa$ describes how 
the distribution concentrates about $\bar{\theta}$. Compared with (\ref{eq:vmf}), we can define two similar parameters to describe the statistical characterization of MFDs on $\mathcal{S}_{M_0 }(w_0)$ as follows.
\begin{define}\label{def:para_axies}
    Consider a random attitude $R\sim \mathcal{M}(F)$. If the probability of elements belonging to $\mathcal{S}_{M_0 }(w_0)$ can be rewritten as 
    \begin{equation*}
        p(R(\theta)) \propto e^{\kappa \cos(\theta - \bar{\theta})},
    \end{equation*}
then $\bar{\theta}$ is referred to as the mean angle of  $\mathcal{M}(F)$ on $\mathcal{S}_{M_0 }(w_0)$ and $\kappa$ is referred to as the concentration parameter of $\mathcal{M}(F)$ on $\mathcal{S}_{M_0 }(w_0)$.
\end{define}
According to Theorem \ref{the:subset_1}, the mean angle and concentration parameter of  $\mathcal{M}(F)$ on $\mathcal{S}_{M_0 }(w_0)$ such that $M\in\mathcal{S}_{M_0 }(w_0)$ are given by
\begin{equation*}
     \bar{\theta} = {w_0 \cdot  \log(M_0M^T)^{\vee}}\ and\ \kappa = \text{tr}(K) - w_0^TKw_0.
\end{equation*}
When $w_0$ is the $i$-th principal axis expressed in $\mathcal{F}_B$, i.e., $w_0 = Ve_i$, we have
    \begin{align*}
        \kappa &= \text{tr}(K) - w_0^TKw_0 \nonumber\\
               &= s_i +s_j +s_k - e_i^TV^TVSV^TVe_i = s_j +s_k,
    \end{align*}
    where $\{i,j,k \} = \{1,2,3 \}$. Moreover, if $M_0 = M$, (\ref{eq:derive_vmf}) can be rewritten as
    \begin{equation*}
        p(R(\theta)) = \frac{e^{s_i}}{c(F)} \exp \left[ \left( s_j +s_k \right)\cos(\theta) \right],
    \end{equation*}
which is the same as (42) in \cite{leeBayesianAttitudeEstimation2018}. Therefore, Definition \ref{def:para_axies} can be viewed as a generalization of the geometric interpretation about MFDs in \cite{leeBayesianAttitudeEstimation2018} and the parameters defined in Definition \ref{def:para_axies} can describe the dispersion of attitude generated by the mean attitude rotating around any specific axis $w_0 \in \mathbb{S}^2$, instead of merely principal axes.

A similar conclusion can be drawn for CGDs.
\begin{theorem}
    Consider a random attitude $R\sim \mathcal{N}_L(M,P)$. Let $\mathcal{S}_{M_0 }(w_0)$ be a subset of $SO(3)$ defined as (\ref{eq:def_subset}). If $M\in \mathcal{S}_{M_0 }(w_0)$, for $R(\theta)\in  \mathcal{S}_{M_0 }(w_0)$, the probability density can be rewritten as
    \begin{equation}\label{eq:one_d_CGD}
        p(R (\theta))\propto \exp\left(-\frac{1}{2} (\theta - \bar{\theta})w_0^T P^{-1} w_0(\theta - \bar{\theta}) \right),
    \end{equation}
    where $\bar{\theta} ={w_0 \cdot  \log(M_0M^T)^{\vee}}$.
\end{theorem}
\begin{proof}
    For $R (\theta)\in  \mathcal{S}_{M_0 }(w_0)$, we have $    R  (\theta)= M_0 \exp \left( \theta w_0^{\wedge} \right)$. Because $M$ also belongs to $\mathcal{S}_{M_0}(w_0)$, an arbitrary $R (\theta)$ belonging to $\mathcal{S}_{M_0 }(w_0)$ can be rewritten as
    \begin{align}\label{ga_M2M_0}
        R (\theta) &= M_0 \exp (\theta \omega_0 ^{\wedge})= M \exp (\bar{\theta} \omega_0 ^{\wedge})^T \exp (\theta \omega_0 ^{\wedge}) \nonumber \\
          &= M \exp \left((\theta-\bar{\theta}) \omega_0 ^{\wedge}\right).
    \end{align}
  
    Substituting $R (\theta) =  M \exp \left((\theta-\bar{\theta}) \omega_0 ^{\wedge}\right)$ into the PDF of $R$, we have
    \begin{equation*}
        p(R(\theta))\propto \exp\left(-\frac{1}{2}  (\theta - \bar{\theta})w_0^T P^{-1} w_0(\theta - \bar{\theta}) \right),
    \end{equation*}
    which verifies (\ref{eq:one_d_CGD}).
\end{proof}
The density function (\ref{eq:one_d_CGD}) of CGDs on the subset $\mathcal{S}_{M_0 }(w_0)$ is similar to the one-dimensional Gaussian distribution of $\theta \in \mathbb{R}$ with mean $\bar{\theta}$ and  variance $1/(w_0^T P^{-1} w_0)$. Thus the following two parameters are defined to characterize the CGD on $\mathcal{S}_{M_0 }(w_0)$.
\begin{define}\label{de:ga_para_def}
Consider a random attitude $R\sim \mathcal{N}_L(M,P)$. If the probability of elements belonging to  $\mathcal{S}_{M_0 }(w_0)$ can be rewritten as
\begin{equation}
     p(R(\theta)) \propto \exp\left[-\frac{(\theta -\bar{\theta})^2}{2\sigma^2} \right],
\end{equation}
where $R(\theta)\in \mathcal{S}_{M_0 }(w_0)$. Then $\bar{\theta}$ is referred to as the mean angle of  $ \mathcal{N}_L(M,P)$ on $\mathcal{S}_{M_0 }(w_0)$ and $1/\sigma^2$ is referred to as the concentration parameter of  $ \mathcal{N}_L(M,P)$ on $\mathcal{S}_{M_0 }(w_0)$.
\end{define}

\begin{remark}\label{re:appro}
    Although the two sets of parameters in Definitions \ref{def:para_axies} and \ref{de:ga_para_def} are defined via different methods, they have similar statistical meanings and magnitudes. In the density functions (\ref{eq:MF2vMF}) and (\ref{eq:one_d_CGD}), the mean angles maximize the probability density on the subset and the concentration parameters describe how distributions on the subset concentrates about 
the mean angle. Besides, using Taylor's formula, the density functions (\ref{eq:MF2vMF}) can be rewritten as
\begin{align*}
      p(R(\theta)) &\propto \exp \left[\kappa\cos(\theta-\bar{\theta}) \right]\\
                                   &\propto \exp \left[- \frac{\kappa (\theta-\bar{\theta})^2}{2} + \mathcal{O}(\Vert \theta-\bar{\theta}  \Vert^4) \right ]
\end{align*}
which indicates that the two concentration parameters, $\kappa$ and $1/\sigma^2$, are almost equivalent when the tails of distributions are small. Therefore, the two sets of parameters are comparable. 
\end{remark}

A numerical example is presented to illustrate these parameters more clearly. Consider a subset $\mathcal{S}_{M_0 }(w_0)$ with $M_0 = I_3$ and $w_0 = 1/\sqrt{2}\ [0, 1, 1]^T$ and three random attitudes characterized by the following MFDs with different parameters of $F_1 = 50I_3,\ F_3 = \exp{(\pi/4 \omega_0^{\wedge})} 50I_3$ and
\begin{gather*}
 F_2  = \left[ {\begin{array}{*{20}c}
   {2} & {0} & {0}  \\
   {0} & {74} & {72}  \\
   {0} & {72} & {74}  \\
 \end{array} } \right],  
\end{gather*}
respectively. The traces of $F_1$, $F_2$, and $F_3$ are the same, but their mean angles $\bar{\theta}_1,\ \bar{\theta}_2,\ \bar{\theta}_3$ and concentration parameters $\kappa_1,\ \kappa_2,\ \kappa_3$ are different. Note that $\kappa_1 = \kappa_3 = 100$, while $\kappa_2 = 4$. Since $\kappa_2 \leq \kappa_1 = \kappa_3$, the distribution $\mathcal{M}(F_2)$ on $\mathcal{S}_{M_0 }(w_0)$ is more dispersed than the other two distributions. Besides, $\bar{\theta}_1 = \bar{\theta}_2 = 0$, while $\bar{\theta}_3 = \pi/4$, which implies that the mean angle of $\mathcal{M}(F_3)$ offsets $45^{\circ}$ from those of $\mathcal{M}(F_1)$ and $\mathcal{M}(F_2)$. These characteristics are also shown in Fig. \ref{fig:1}.

\begin{figure}[htbp]
	\begin{center}
		\includegraphics[width=\linewidth]{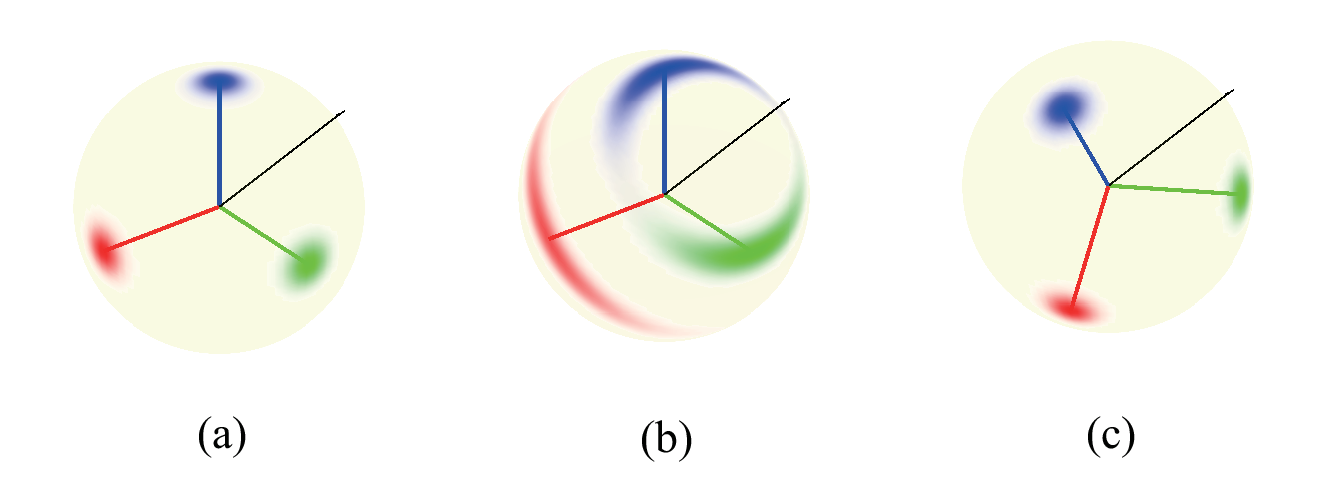}
		\caption{Matrix Fisher distributions with different parameters. The marginal distribution for each column of $R$ is shown on the unit sphere as red, green and blue shades, as presented in \cite{leeBayesianAttitudeEstimation2018}. The red, green and blue lines represent the first, second, and third columns of the mean attitude, respectively, and the black line represents the axis $w_0$. (a) $\mathcal{M}(F_1)$: $\kappa_1 = 100,\ \bar{\theta}_1=0$; (b) $\mathcal{M}(F_2)$: $\kappa_2 = 4,\ \bar{\theta}_2=0$; (c) $\mathcal{M}(F_3)$: $\kappa_3 = 100,\ \bar{\theta}_3=\pi/4$}   \label{fig:1}
	\end{center}
 
 \end{figure}

\subsection{The evolutions of MFDs and CGDs along Bayes' rule}


Next, the evolutions of MFDs and CGDs along Bayes’ rule are derived via the evolutions of the characteristic parameters defined in Definitions \ref{def:para_axies} and \ref{de:ga_para_def}. The evolution of MFDs is given first as follows.

\begin{theorem}\label{th:poster_anly}
Let the prior attitude distribution be characterized by $R\sim \mathcal{M}(F^-)$ and the likelihood distribution by $M^{m}|R\sim \mathcal{M}(RK^m)$ with $K^m \in \mathbb{R}^{3\times 3}$ and $K^m =\left( K^m \right)^T$. Denote the proper right polar decomposition of $F^-$ by $F^- = M^- K^-$ and a single sample of  $\mathcal{M}(RK^m)$ by $\tilde{M}^m$. The posterior probability density of  $R(\theta) \in \mathcal{S}_{M^{-} }(w_0)$ can be rewritten as
    \begin{equation}\label{eq:subset_3}
        p\left(R(\theta)|\tilde{M}^m\right) \propto  \exp \left[ \kappa^{+} \cos(\theta-\bar{\theta}^{+}) \right],
    \end{equation}
    where $\mathcal{S}_{M^- }(w_0)$ containing both $M^-$ and $\tilde{M}^m$ is given by Proposition \ref{the:subset_2}, and $\kappa^{+}$ and $\bar{\theta}^{+}$ are given by
     \begin{subequations}
        \begin{equation}\label{eq:kappa_3}
             \kappa^{+} = \sqrt{(\kappa^{-})^2 +(\kappa^{m})^2 + 2\kappa^{-} \kappa^{m} \cos(\bar{\theta}^{m})},
        \end{equation}
        \begin{equation}\label{eq:theta_3_MF}
            \bar{\theta}^+ = {\rm sign}(\bar{\theta}^m) \arccos\left( \frac{\kappa^- + \kappa^m\cos(\bar{\theta}^m)}{\kappa^+}\right).
        \end{equation}
    \end{subequations}
In (\ref{eq:kappa_3}) and (\ref{eq:theta_3_MF}), $\bar{\theta}^{-}$ and $\bar{\theta}^{m}$ are the mean angles of  $\mathcal{M}(F^-)$ and $\mathcal{M}(F^m)$ respectively, and $\kappa^-$ and $\kappa^m$ are the concentration parameters of  $\mathcal{M}(F^-)$ and $\mathcal{M}(F^m)$ respectively, where $F^m=\tilde{M}^mK^m$. It follows from the definition of $\mathcal{S}_{M^- }(w_0)$ that $\bar{\theta}^{-}=0$. 
\end{theorem}
\begin{proof}
    For $R\in \mathcal{S}_{M^- }(w_0)$, we have $R = M^- \exp (\theta \omega_0 ^{\wedge}).$
    Then, using Bayes' rule, we have
    \begin{align*}
       & p(R(\theta)|\tilde{M}^m) \propto p(R)p(M^m=\tilde{M}^m|R)\\
       &= \text{etr}\left( (K^- (M^-)^T + K^m (M^m)^T) M^- \exp (\theta \omega_0 ^{\wedge}) \right) \nonumber\\
       &= \text{etr}\left( (K^-  + K^m (M^m)^TM^-) \exp (\theta \omega_0 ^{\wedge}) \right). \nonumber\\
    \end{align*}
    According to Proposition \ref{the:subset_2}, $\tilde{M}^m \in \mathcal{S}_{M^- }(w_0)$ and thus we have $(\tilde{M}^m)^TM^- = \left( \exp (\bar{\theta}^m \omega_0 ^{\wedge}) \right)^T.$
    Therefore, the probability density of the elements belonging to $\mathcal{S}_{M_0 }(w_0)$ can be rewritten as
        \begin{align*}
        &p(R(\theta)|\tilde{M}^m) \propto \text{etr}\left( (K^-  + K^m (M^m)^TM^-) \exp (\theta \omega_0 ^{\wedge}) \right) \nonumber\\
               &= \text{etr}\left( K^-\exp (\theta \omega_0 ^{\wedge}) + K^m\exp ((\theta-\bar{\theta}^m) \omega_0 ^{\wedge})\right) \nonumber\\
                &\propto {\text{exp}}\left[ {\left( {{\text{tr}}(K^- ) - w_0^T K^- w_0 } \right)\cos(\theta ) + } \right.\nonumber \\
                &\left. {\left( {{\text{tr}}(K^m ) - w_0^T K^m w_0 } \right)\cos(\theta  - \bar{\theta}^m )} \right] \nonumber \\ 
                &= \text{exp} \left[ \left(\kappa^- + \kappa^m \cos(\bar{\theta}^m)\right) \cos(\theta) + \kappa^m \sin(\bar{\theta}^m) \sin(\theta) \right]\nonumber\\
                &= \text{exp} \left[ \kappa^+\cos(\theta-\bar{\theta}^+) \right],
    \end{align*}
where $\kappa^+$ and $\bar{\theta}^+$ are given by the induction formula of the cosine function. This verifies the conclusion.
\end{proof}
The above theorem indicates how the prior and the likelihood distributions affect the posterior distribution on $\mathcal{S}_{M^-}(w_0)$ when the underlying prior and likelihood distributions are approximated by MFDs. Theorem \ref{th:poster_anly} can be applied to the filtering procedure of any Bayesian attitude filter using MFDs to represent attitude uncertainty.

A similar theorem is given for the case that the underlying prior and likelihood distributions are approximated by CGDs.
\begin{theorem}\label{th:poster_anly_ga}
    Consider the prior attitude distribution characterized by $R\sim \mathcal{N}_L(M^{-},P^{-})$ and the likelihood distribution by $M^{m}|R\sim \mathcal{N}_L(R,P^{m})$. Denote a single  sample of  $\mathcal{N}_L(R,P^{m})$ by $\tilde{M}^m \in SO(3)$. The posterior probability density of  $R(\theta) \in \mathcal{S}_{M_0 }(w_0)$ can be rewritten as
    \begin{equation}\label{eq:one_d_ga_post}
        p(R(\theta)|\tilde{M}^{m}) \propto \exp\left[-\frac{(\theta -\bar{\theta}^{+})^2}{2(\sigma^{+})^2} \right],
    \end{equation}
The subset $\mathcal{S}_{M^- }(w_0)$ containing $M^-$ and $\tilde{M}^{m}$ is given by Proposition \ref{the:subset_2}, and $1/(\sigma^{+})^2$ and $\bar{\theta}^{+}$ are given by
    \begin{subequations}
        \begin{equation}\label{eq:kappa_ga_3}
            \frac{1}{(\sigma^{+})^2} = \frac{1}{(\sigma^{-})^2}+\frac{1}{(\sigma^{m})^2}
        \end{equation}
        \begin{equation}\label{eq:theta_ga_3}
            \bar{\theta}^{+} = \frac{(\sigma^{-})^2}{(\sigma^{m})^2 + (\sigma^{-})^2}\bar{\theta}^{m}.
        \end{equation}
    \end{subequations}
    In (\ref{eq:kappa_ga_3}) and (\ref{eq:theta_ga_3}), $\bar{\theta}^{-}$ and $\bar{\theta}^{m}$ are the mean angles of  $\mathcal{N}_L(M^{-},P^{-})$ and $\mathcal{N}_L(\tilde{M}^{m},P^{m})$ respectively and $1/(\sigma^{-})^2$, and $1/(\sigma^{m})^2$ are the concentration parameters of  $\mathcal{N}_L(M^{-},P^{-})$ and $\mathcal{N}_L(\tilde{M}^{m},P^{m})$ respectively. It follows from the definition of $\mathcal{S}_{M_0 }(w_0)$ that $\bar{\theta}^{-}=0$.
\end{theorem}
\begin{proof}
    According to Bayes' rule and the definition of CGDs, the posterior attitude distribution is given by
    \begin{align*}
        &p(R|\tilde{M}^{m}) \propto p(R)p(M^m =\tilde{M}^{m}|R) \\
        &\propto \exp \left\{-\frac{1}{2}\left[ \log \left((M^-)^TR\right) ^{\vee}\right]^T\! (P^{-})^{-1} \! \left[ \log \left((M^-)^TR\right) ^{\vee}\right] \right\}  \times\\
        &\exp \left\{ -\frac{1}{2}\left[ \log \left((\tilde{M}^{m})^TR\right) ^{\vee}\right]^T\!  (P^{m})^{-1}\!  \left[ \log \left((\tilde{M}^{m})^TR\right) ^{\vee}\right] \right\}
    \end{align*}
Letting $R(\theta) \in \mathcal{S}_{M^- }(w_0)$ and noting $M^-, \tilde{M}^{m} \in  \mathcal{S}_{M^- }(w_0)$ , we have
    \begin{align*}
         \log\left((M^-)^TR(\theta)\right) ^{\vee} &= \theta w_0, \\
        \log \left((\tilde{M}^{m})^T R(\theta)\right) ^{\vee} &=\log \left((\tilde{M}^{m})^T (M^-) (M^-)^TR(\theta)\right) ^{\vee}\\
       & = w_0\left(\theta - \bar{\theta}^m \right).
    \end{align*}
    The posterior distribution can then be rewritten as
   \begin{align*}
        &p(R(\theta)|\tilde{M}^{m})\\
        &\propto  \exp \left\{-\frac{1}{2} \theta w_0^T (P^{-})^{-1} w_0 \theta\right\}  \times\\
        &\exp \left\{-\frac{1}{2} \left(\theta - \bar{\theta}^m \right)w_0^T (P^{m})^{-1} w_0\left(\theta - \bar{\theta}^m \right) \right\}\\
       &\propto\exp \left\{-\frac{\theta^2}{2(\sigma^-)^2}\right\}\exp \left\{-\frac{(\theta-\bar{\theta}^m)^2}{2(\sigma^m)^2}\right\}\propto C \exp \left\{-\frac{(\theta-\bar{\theta}^+)^2}{2(\sigma^+)^2}\right\},
    \end{align*}
    where $\bar{\theta}^+$ and ${\sigma}^+$ can be shown to be the same as (\ref{eq:kappa_ga_3}) and (\ref{eq:theta_ga_3}). Therefore, the conclusion is verified.
\end{proof}
The parameters $1/(\sigma^2)^+$ and $\theta^+$ characterize the posterior distribution approximated by CGDs on $\mathcal{S}_{M_1}(w_0)$, and their expressions with respect to $1/(\sigma^-)^2,\theta^-,\ 1/(\sigma^m)^2 $ and $\theta^m$ indicate the impact of the prior and the likelihood distributions on the posterior distribution.

\subsection{Comparisons and discussions}\label{sec:c_and_d}
The results given by Theorems \ref{th:poster_anly} and \ref{th:poster_anly_ga} are summarized in Table \ref{tab:-1} and the approximation introduced in Remark \ref{re:appro} is used such that the parameters of both MFDs and CGDs are expressed with $\kappa$ and $\bar{\theta}$. The trends of $\kappa^+$ as a function of $\Delta \bar{\theta} = \bar{\theta}^m - \bar{\theta}^-$, and $\bar{\theta}^+$ as a function of  $\kappa^m/\kappa^-$ with different $\Delta \bar{\theta}$ are shown in Figs. \ref{Fig:posterior_kappa} and \ref{Fig:posterior_theta}. Since the parameters are defined on $\mathcal{S}_{M^-}(w_0)$, we have $\bar{\theta}^- = 0$ and $\Delta \bar{\theta} = \bar{\theta}^m $.
 \begin{table}[htbp]
\centering
\caption{Mean angles and concentration parameters of distributions on $\mathcal{S}_{M^-}(w_0)$}
\label{tab:-1}
\renewcommand{\arraystretch}{2.5} 
\begin{tabular}{c|cc}
\hline
             &MFDs&CGDs\\ \hline
 $\kappa^+$&$\sqrt{(\kappa^-)^2 + (\kappa^m)^2 + 2\kappa^- \kappa^m\cos(\Delta \bar{\theta})}$&$\kappa^- + \kappa^m$\\ \hline
                       $\bar{\theta}^+$&$\text{sign}(\Delta \bar{\theta}) \arccos\left(\displaystyle\frac{\kappa^- + \kappa^+\cos(\Delta \bar{\theta})}{\kappa^+}\right)$&$\displaystyle\frac{\kappa^m}{\kappa^- + \kappa^m} \Delta \bar{\theta}$\\ \hline
\end{tabular}
\end{table}
\begin{figure}[htbp]
	\begin{center}
		\includegraphics[width=\linewidth]{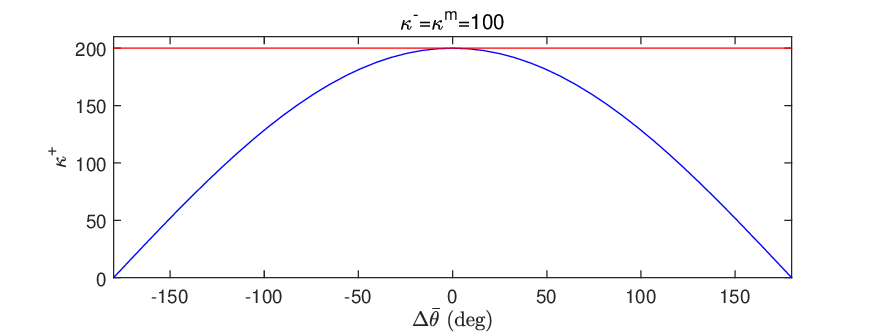}
		\caption{The trend of $\kappa^+$ as a function of $\Delta \bar{\theta}$. Blue: $\kappa^+$ given by MFD-based Bayesian filters. Red: $\kappa^+$ given by CGD-based Bayesian filters}\label{Fig:posterior_kappa}
	\end{center}
 \end{figure}
\begin{figure}[htbp]
	\begin{center}
		\includegraphics[width=\linewidth]{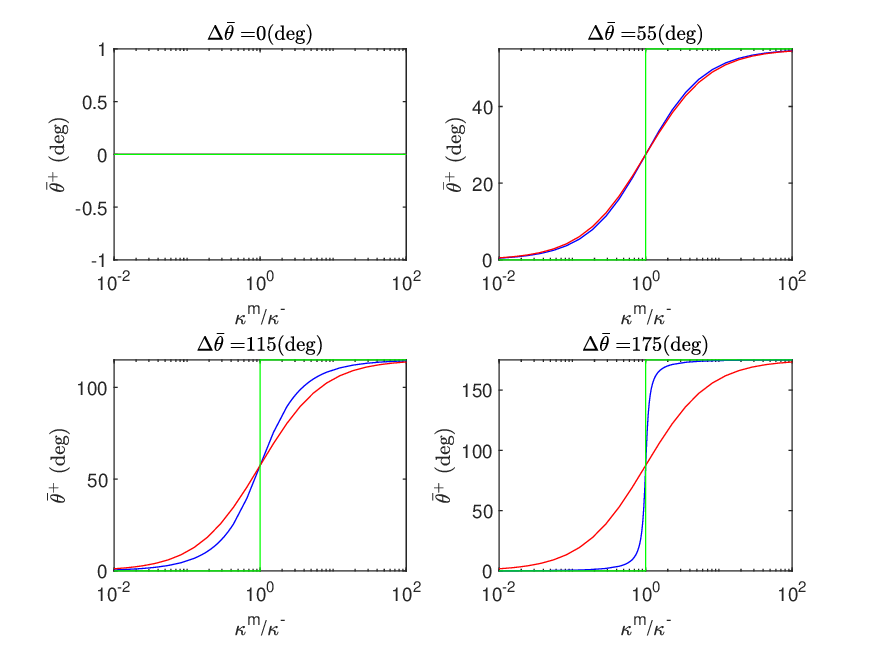}
		\caption{The trends of $\bar{\theta}^+$ on $\kappa^m/\kappa^-$ with different $\Delta{\bar{\theta}}$. Blue: $\bar{\theta}^+$ given by MFD-based Bayesian filters. Red: $\bar{\theta}^+$ given by CGD-based Bayesian filters. Green: the mean angle with higher confidence in $\bar{\theta}^-$ and $\bar{\theta}^m$} \label{Fig:posterior_theta}
	\end{center}
 \end{figure}

In this subsection, the parameters given by MFD-based Bayesian filters are denoted by $\kappa^+_{MFD}$ and $\bar{\theta}^+_{MFD}$, while the parameters given by CGD-based Bayesian filters are denoted by $\kappa^+_{CGD}$ and $\bar{\theta}^+_{CGD}$. Comparing $\kappa^+_{MFD}$ and $\kappa^+_{CGD}$, it can be seen from Table \ref{tab:-1} and Fig. \ref{Fig:posterior_kappa} that they are almost the same except for the term $\cos(\Delta \bar{\theta})$. Meanwhile, the closer $|\Delta \bar{\theta}|$ is to $\pi$, the closer $\kappa^+_{MFD}$ is to 0, i.e., the closer the posterior distribution on $\mathcal{S}_{M_1}(w_0)$ is to a uniform distribution. Therefore, the uncertainty in the posterior estimation provided by the Bayesian attitude filter with MFDs increases when there is a large discrepancy between the mean attitudes of the prior and the likelihood. In contrast, the uncertainty in the attitude filter with CGDs does not exhibit such an increase, since $\kappa^+_{CGD}$ is independent of $\Delta \bar{\theta}$. Besides, $\kappa^+_{MFD}$ and $\kappa^+_{CGD}$ are the same when $\Delta \bar{\theta} = 0$, which implies MFD-based and CGD-based filters give similar uncertainty descriptions around the axis $w_0$ when the mean attitude of the prior and the likelihood distributions are consistent.

Meanwhile, $\bar{\theta}^+_{MFD}$ and $\bar{\theta}^+_{CGD}$ are also different. Fig. \ref{Fig:posterior_theta} shows that $\bar{\theta}^+_{MFD}$ is always closer to the mean angles with a higher confidence, i.e., with larger concentration parameter, in $\bar{\theta}^-$ and $\bar{\theta}^m$. This difference is more significant when $\Delta \bar{\theta} $ is larger. The following presents an explanation from a theoretical perspective. 
\begin{proposition}\label{prop:theta_de}
    Define $\Delta\theta^+ \in \mathbb{R}$ as 
    \begin{align}
        &\Delta\bar{\theta}^+ (k,\Delta\bar{\theta})=\bar{\theta}^+_{MFD} -\bar{\theta}^+_{CGD} \nonumber\\
        &= {\rm sign}(\Delta\bar{\theta})\arccos\left( \frac{1 + k\cos\Delta\bar{\theta}}{\sqrt{1 + k^2 + 2k \cos(\Delta\bar{\theta})}}\right) - \frac{k}{1 + k}\Delta\bar{\theta},
    \end{align}
    where $k= {\kappa^m}/{\kappa^-}$.  Then the following property holds:
 \begin{subequations}\label{eq:daxiao_0p}
     \begin{gather}\label{eq:daxiao_1}
        \Delta\bar{\theta}^+ (k>1,\ 0 < \Delta\bar{\theta} \leq \pi ) > 0,
     \end{gather}
     \begin{gather}\label{eq:daxiao_2}
        \Delta\bar{\theta}^+ (0<k<1,\ 0 < \Delta\bar{\theta} \leq \pi ) < 0,
     \end{gather}
          \begin{gather}\label{eq:daxiao_3}
        \Delta\bar{\theta}^+ (k>1,\ -\pi \leq \Delta\bar{\theta} < 0 ) < 0,
     \end{gather}
     \begin{gather}\label{eq:daxiao_4}
        \Delta\bar{\theta}^+ (0<k<1,\ -\pi \leq \Delta\bar{\theta} < 0 ) > 0,
     \end{gather}
     \begin{gather}\label{eq:daxiao_5}
                \lim_{k \rightarrow +\infty}\Delta\bar{\theta}^+ = \Delta\bar{\theta}^+ (k=0,\ 1) = 0.
     \end{gather}
 \end{subequations}
\end{proposition}
\begin{proof}
    It is obvious that $\Delta\bar{\theta}^+ (k, \Delta\bar{\theta} = 0  ) \equiv 0$ for any $k>0$. Then, for $0 \leq \Delta\bar{\theta} \leq \pi $, the partial derivative of $\Delta\bar{\theta}^+ (k,\Delta\bar{\theta})$ with respect to $\Delta\bar{\theta}$ is given by
\begin{align*}
    \frac{\partial }
{{\partial \Delta\bar{\theta}}}\Delta\bar{\theta}^+ \left( {k,\Delta\bar{\theta}} \right) = \frac{{k\left( {k - 1} \right)\left( {1 - \cos \bar \Delta\bar{\theta} } \right)}}
{{\left( {1 + k^2  + 2k\cos \bar \Delta\bar{\theta} } \right)\left( {1 + k} \right)}}. \\ 
\end{align*}
Therefore, invoking $\Delta\bar{\theta}^+ (k, \Delta\bar{\theta} = 0  ) \equiv 0$, we have
\begin{align*}
     k < 1 \Rightarrow \frac{\partial }
{{\partial \Delta\bar{\theta} }}\Delta\bar{\theta}^+   < 0 \Rightarrow \Delta\bar{\theta}^+  \left( {k < 1,0 \leq \theta  \leq \pi } \right) < 0, \\ 
  k > 1 \Rightarrow \frac{\partial }
{{\partial \Delta\bar{\theta} }}\Delta\bar{\theta}^+   > 0 \Rightarrow \Delta\bar{\theta}^+  \left( {k > 1,0 \leq \theta  \leq \pi } \right) > 0, \\ 
\end{align*}
verifying (\ref{eq:daxiao_1}) and (\ref{eq:daxiao_2}) respectively. Next, substituting $k=1$ into $\Delta\bar{\theta}^+ (k,\Delta\bar{\theta})$, we have 
\begin{align*}
   &\Delta\bar{\theta}^+ (k = 1) = \text{sign}(\Delta\bar{\theta}) \arccos\left( \frac{1 + \cos\Delta\bar{\theta}}{\sqrt{2 + 2 \cos(\Delta\bar{\theta})}}\right) - \frac{1}{2}\Delta\bar{\theta} \\
        &=\text{sign}(\Delta\bar{\theta}) \arccos\left( \frac{2\cos^2(\Delta\bar{\theta}/2)}{\sqrt{4\cos^2(\Delta\bar{\theta}/2)}}\right) - \frac{1}{2}\Delta\bar{\theta} = 0,\\
        &\Delta\bar{\theta}^+ (k = 0) = \text{sign}(\Delta\bar{\theta}) \arccos\left(1\right) - 0 = 0,\\
            &\lim_{k \rightarrow +\infty}\Delta\bar{\theta}^+ = \text{sign}(\Delta\bar{\theta}) \arccos\left( \cos\Delta\bar{\theta} \right) - \Delta\bar{\theta} = 0,
\end{align*}
verifying (\ref{eq:daxiao_5}). Besides, because $\Delta\bar{\theta}^+ (k, \Delta\bar{\theta} )$ is an odd function with respect to $\Delta\bar{\theta}$,  (\ref{eq:daxiao_1}) and (\ref{eq:daxiao_2}) can be easily extended to $\Delta\bar{\theta} \in [-\pi,\pi]$, which shows (\ref{eq:daxiao_3}) and (\ref{eq:daxiao_4}). Therefore, the conclusion is verified.
\end{proof}
When $k>1$, we have $\kappa^m > \kappa^-$, which means that measurements are more reliable than the prior information on $\mathcal{S}_{M^-}(w_0)$. According to (\ref{eq:daxiao_0p}), it is concluded that $\bar{\theta}^+_{MFD}$ is closer to $\bar{\theta}^m$ than $\bar{\theta}^{+}_{CGD}$. When $k<1$,  $\bar{\theta}^+_{MFD}$ is closer to $\bar{\theta}^-$ (i.e., 0) than $\bar{\theta}^+_{CGD}$. In a word, compared with the attitude filter with 
CGDs, the mean angle estimated by the Bayesian attitude filter with MFDs is always closer to the mean angle of a higher-confidence distribution (either the likelihood or the prior). Additionally, according to (\ref{eq:daxiao_3}), the mean angles given by the two attitude filters are almost the same when $\kappa^m$ and $\kappa^-$  are the same or one is significantly larger than the other.

Since $\kappa^+$ depends on $\Delta\bar{\theta}$, the attitude filters with MFDs can detect the discrepancy between the predicted attitude and measurements, thereby quickly filtering out the negative effects of large initial errors by adjusting the uncertainty of attitude estimates. Meanwhile, according to Proposition \ref{prop:theta_de}, MFD-based filters can filter out low-confidence information fast, which also endows MFD-based filters with robustness to initial errors.  In Section \ref{sec:single_axis}, moreover, the two properties are related to the convergence rate of the filters with MFDs for single-axis rotation. These properties arise from the generic Bayesian fusion mechanism and are independent of the propagation of the distribution along the kinematics or dynamics, which motivates the fast MFD-based filter in the next section.


\subsection{Numerical examples}
To illustrate the difference in Bayesian attitude filters with MFDs and CGDs more clearly, we present a numerical example. The initial prior distribution at $t=0$ is set as $R_0 \sim \mathcal{M}(F_0)$ with $F_0 = 55\exp((35\pi/36)w_0^{\wedge})$ and $w_0 = [0.54,0.54,0.65]^T$, and the corresponding CGD is chosen as $R_0 \sim \mathcal{N_L} (\exp((35\pi/36)w_0^{\wedge}),P_0)$ with $P_0 = 0.0091I_3$. The initial prior distribution is fused with the likelihood distribution twice via Bayes' rule at $t=1$ and $t=2$, and the posterior distribution at $t=1$ is chosen as the prior distribution at $t=2$. Since computing the full posterior distributions of Bayesian filters with CGDs is nontrivial, approximated posterior distributions are obtained with the first-order approximation of Baker-Campbell-Hausdorff formula, i.e. $\log\left(\exp(\xi^{\wedge})\exp(\zeta^{\wedge})\right) \approx \xi^{\wedge} +\zeta^{\wedge}$ for $\xi,\ \zeta \in \mathbb{R}^3$ and $\Vert \xi \Vert,\ \Vert \zeta \Vert <<  1$. The likelihood distribution is set as  $R_m |R_t \sim \mathcal{M}(R_tK_m)$ with $K_m = 60I_3$ for MFDs and $R_m |R_t \sim \mathcal{N_L} (R_t,P_m)$ with $P_m = 0.0083 I_3$ for CGDs. The true attitudes are set as $R_t=I_3$ for $t=0,1,2$. These parameter settings imply that the attitude of a stationary rigid body, whose body-fixed frame coincides with the inertial frame,  is estimated with a large initial error around the axis $w_0$ and direct attitude measurements. The measurements are independent of the initial error and therefore have higher confidence levels. The results are summarized in Table \ref{tab:0} and Fig. \ref{fig:0}.

At $t=0$, the prior and the likelihood distributions fed into the MFD-based and the CGD-based attitude filters are almost the same. However, after the first fusion at $t=1$, the posterior distributions of the two filters are noticeably different, as shown in Figs. \ref{fig:0}b and \ref{fig:0}e. Their differences are twofold. Firstly, although the prior and the likelihood distributions are both highly concentrated about the mean attitude, the posterior distribution given by the MFD-based attitude filter is significantly uncertain for rotations about the axis $w_0$ and $\kappa$ also steeply decreases at $t=1$. By contrast, the posterior distribution given by the CGD-based attitude filter is even more concentrated around the mean attitude than the prior and the likelihood distributions, as shown in Fig. \ref{fig:0}f and Table \ref{tab:0}. As discussed in Section \ref{sec:c_and_d}, the reason for this difference is that the concentration parameter of the posterior CGD is independent of the mean angle, while that of MFDs is not. Secondly, the mean attitude of the MFD-based filter at $t=1$ is closer to the mean attitude of the likelihood distribution than that of the CGD-based filter at $t=1$, which corresponds to Proposition \ref{prop:theta_de}. Then, at $t=2$, the posterior distribution given by the MFD-based attitude filter is almost the same as the likelihood distribution because the prior distribution is highly uncertain for rotations about the axis $w_0$, as shown in Fig. \ref{fig:0}c. In contrast, as shown in Fig. \ref{fig:0}f, the error of the CGD-based attitude filter is still large.

 \begin{table}[htbp]
\centering
\caption{Mean angles and concentration parameters of distributions on $\mathcal{S}_{I_3}(w_0)$}
\label{tab:0}
\begin{tabular}{c|c|cccc}
\hline
Probabilistic Model&             &Likelihood&$t=0$& $t=1$& $t=2$\\ \hline
MFDs& $\kappa$&120&110&14.16&123.71\\
                      & $\bar{\theta}$&0°&175°&42.62°&6.43°\\ \hline
CGDs& $1/\sigma^2$&120&110&230&350\\
                      & $\bar{\theta}$&0°&175°&83.70°&55°\\ \hline
\end{tabular}
\end{table}
 
\begin{figure}[htbp]
	\begin{center}
		\includegraphics[width=\linewidth]{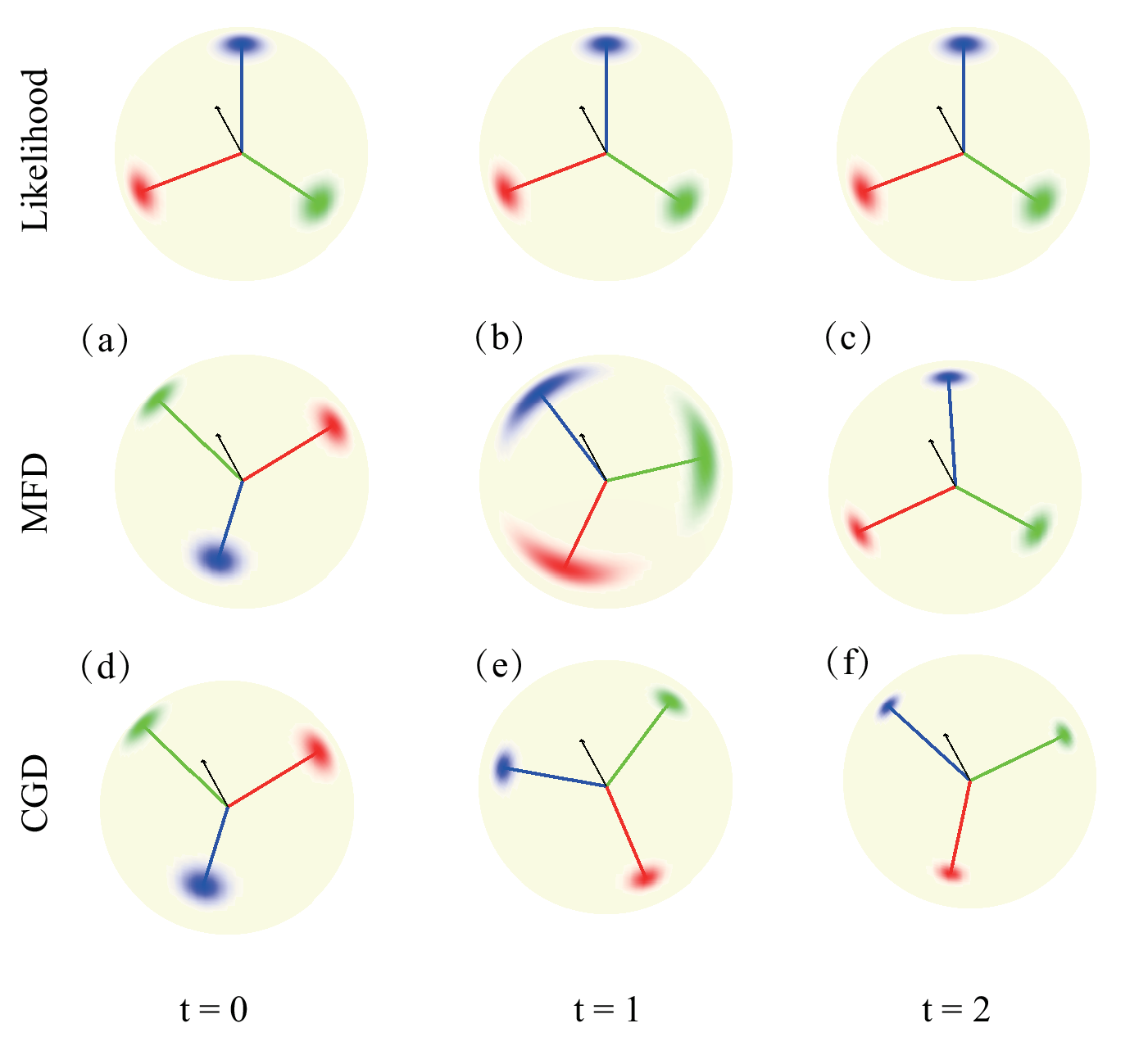}
		\caption{The posterior distribution with several direct attitude measurements.}\label{fig:0}
	\end{center}
 \end{figure}

\section{MFD-Based Fast Nonlinear Filters for Attitude Estimation}
\label{sec:Filter}
In this section, we apply MFDs to characterize attitude error on $SO(3)$ and thereby  propose two MFD-based fast nonlinear filters for attitude estimation. 

\subsection{Attitude estimation error following MFDs}
We construct the following probability model to characterize the uncertainty of random attitudes with invariant attitude error following MFDs on $SO(3)$.

\begin{define}
   The discrepancy of the random matrix $R \in SO(3)$ with respect to $\bar{R}\in SO(3)$ is characterized by the left- or right-invariant error with MFDs on $SO(3)$if it satisfies
   \begin{subequations}
        \begin{equation}\label{eq:right_MF}
        R = \delta R \bar{R},{\ }\delta R \sim \mathcal{M}(N),
    \end{equation}
    or
    \begin{equation}\label{eq:left_MF}
        R = \bar{R} \delta R,{\ }\delta R \sim \mathcal{M}(N),
    \end{equation}
   \end{subequations}
where $\bar{R}$ and $ N = N^T \in \mathbb{R}^{3 \times 3}$ are deterministic parameters.
For simplicity,  (\ref{eq:right_MF}) and (\ref{eq:left_MF}) are denoted by $R \sim \mathcal{P_R M} (\bar{R},N)$ and $R \sim \mathcal{P_L M} (\bar{R},N)$, respectively.
\end{define}
The right-invariant error is invariant to right multiplications for any deterministic group element $\bar{R}_R \in SO(3)$, and the left-invariant error has similar invariance for left multiplications. The relationship between the above probability model and MFDs are presented as follows.

\begin{proposition}\label{pro:rela_MF}
    Consider a random rotation matrix $R \sim \mathcal{M}(F)$ with $F\in\mathbb{R}^{3\times 3}$ and the proper left and right decompositions of $F$ are given by $F=K'M'$ and $F=MK$, respectively. Then, the probability model with invariant errors on $SO(3)$ is given respectively by
\begin{equation*}
    R \sim \mathcal{P_R M} (M',K')\ and \ R \sim \mathcal{P_L M} (M,K).
\end{equation*}
\end{proposition}
\begin{proof}
    For $R \sim \mathcal{P_R M} (M',K')$, we have 
    \begin{equation*}
        R = \delta R M',{\ } \delta R \sim \mathcal{M}(K').
    \end{equation*}
    According to (\ref{eq:trans}), we can obtain $R \sim \mathcal{M}(F'),$ where $F' = K'M' = F$.

    Similarly, for $R \sim \mathcal{P_L M} (M,K)$, we have 
    \begin{equation*}
        R = M \delta R,{\ } \delta R \sim \mathcal{M}(K).
    \end{equation*}
    According to (\ref{eq:trans}), we can obtain $R \sim \mathcal{M}(F''),$ where $F'' = MK = F$.
\end{proof}

Let $\delta R = \exp(\xi^{\wedge})$. When $\xi$ is small, or equivalently the random error $\delta R$ is highly concentrated about $I_3$, we can establish a connection between $N$ and the covariance of $\xi$.
\begin{proposition}
\label{pro:GaApp}
   Let the rotation matrix $\delta R$ be parameterized by the exponential map $\exp:\mathfrak{so}(3) \rightarrow SO(3)$ as $\delta R(\xi) = \exp\left(\xi^{\wedge}\right)$. Suppose that $\delta R \sim \mathcal{M}(N)$, and the proper SVD of  $N$ is given by $N = V S V^T$, where $S={\rm diag}[s_1,s_2,s_3]$. If  $s_2+s_3 \gg 0$, then  $\xi$  follows a 3-dimensional Gaussian distribution with
    \begin{equation}
        \xi \sim \mathcal{N}(0,V\left({\rm tr}(S)I_3-S\right)^{-1}V^T)\label{eq:GaApp}
    \end{equation}
\end{proposition}
\begin{proof}
Let $\delta R'=V^T \delta R V$, $P=\left({\rm tr}(S)I_3-S\right)^{-1} = \text{diag}[1/(s_2+s_3), 1/(s_1+s_3), 1/(s_1+s_2)] $, and $\delta R'=\exp({\xi'} ^{\wedge})$.  Replacing $\delta R$ with $\exp(\xi^{\wedge})$, we can obtain
    \begin{align*}
        \exp({\xi'} ^{\wedge})=\delta R'=V^T \exp(\xi^{\wedge}) V =\exp(V^T \xi^{\wedge} V),
    \end{align*}
which implies that $\xi' = V^T \xi$. According to (36) in \cite{leeBayesianAttitudeEstimation2018b}, $\xi'$ follows a Gaussian distribution with $\xi' \sim \mathcal{N}(0,P)$. Using $\xi' = V^T\xi$, we obtain $\xi \sim \mathcal{N}(0,VPV^T)$, which shows (\ref{eq:GaApp}).
\end{proof}
Therefore, $N$ can be uniquely determined by the covariance matrix $P \in \mathbb{R}^{3\times 3}$ of $\xi \sim \mathcal{N}(0,P)$ as
\begin{equation}\label{eq:N2P}
    N = U \left( \frac{1}{2}tr(\Lambda^{-1})I_3 - \Lambda^{-1} \right) U^T,
\end{equation}
where $U$ and $\Lambda$ are given by the proper SVD of $P = U \Lambda U^T$. In the following derivation, the information matrix $\bar{N}=P^{-1}$ is used to describe the uncertainty of $\xi$ in place of the covariance matrix. The replacement avoids the potential singularity caused by the matrix inverse when the uncertainty is such large that $s_2+s_3$ approaches zero in filtering. Correspondingly, $N$ is recovered from $\bar{N}$ by
\begin{equation}\label{eq:N2P_inv}
    N=\frac{1}{2}{\rm tr}(\bar{N})I_3-\bar{N},
\end{equation}
which is equal to \eqref{eq:N2P} when $s_2+s_3>0$ and is free of matrix inverse.

The proposed probability model divides the random rotation matrix $R$ into two mutually independent parts: a deterministic central  attitude $\bar{R}$ and a small random error $\delta R$. They can be separately handled. Additionally, the random error can be propagated and updated by a linearized error system and measurement equations. These characteristics will be utilized in the following attitude filter development.


\subsection{Filtering with the right-invariant error on SO(3)}
\label{sec:right_filter}
Consider the system with discrete attitude kinematics (\ref{eq:AttKiDe}) and measurement model (\ref{eq:AttMea}). The Bayesian filter with the right-invariant error approximates the posterior with (\ref{eq:right_MF}), i.e., we assume that the posterior probability of attitude at $t_k$ is characterized by
\begin{equation}\label{eq:post_assum_L}
    R_k|\{\tilde{b}_k^i\}\sim \mathcal{P_R M}(\hat{R}_k, N_k),
\end{equation}
where $\hat{R}_k$ and $N_k$ are computed at each step as follows.

\subsubsection{Prior}
To obtain the prior attitude distribution, we analyze the deterministic central attitude and random right-invariant error along (\ref{eq:AttKiDe}) separately.

Given $\hat{R}_{k-1}$, $N_{k-1}$ and $\tilde{\omega}_k$, the central attitude at $t_k$, denoted by $\hat{R}_{k|k-1}$, is obtained from the deterministic part of (\ref{eq:AttKiDe}), as
\begin{equation}\label{eq:pro_R}
    \hat{R}_{k|k-1} = \hat{R}_{k-1} \exp\{ (h\tilde{\omega}_k)^{\wedge} \}.
\end{equation}

Assuming that the prior probability of $R_k$ is characterized by $R_k = \mathcal{P_R M}(\hat{R}_{k|k-1}, N_{k|k-1})$, then we have
\begin{equation}\label{eq:err_right}
    R_k = \delta R_{k|k-1} \hat{R}_{k|k-1},\ \delta R_{k|k-1} \sim \mathcal{M}(N_{k|k-1}),
\end{equation}
Next, we compute $N_{k|k-1}$ with the linearized error equation. Assume that $\delta R_{k|k-1}$ is highly concentrated about $I_3$ and $\Vert w_k \Vert$ is small. Denoting by $\xi_{k|k-1} \in \mathbb{R}^3$,  it follows that $\Vert \xi_{k|k-1}  \Vert << 1$  and thus
\begin{equation}\label{eq:appro_R}
    \delta R_{k|k-1} = \exp(\xi_{k|k-1}^{\wedge}) = I + \xi_{k|k-1}^{\wedge} + \mathcal{O}(\Vert \xi_{k|k-1} \Vert^2).
\end{equation}
Substituting (\ref{eq:pro_R}) and (\ref{eq:AttitudeKm}) into  (\ref{eq:err_right}), the estimation error is then derived as
\begin{align}\label{eq:right_err_ki}
    \delta R_{k|k-1} &= R_{k-1} \exp\{ (h \tilde{\omega}_k + w_k)^{\wedge} \} \exp\{ (h \tilde{\omega}_k)^{\wedge} \}^T \hat{R}_{k-1}^T \nonumber \\
                     &\approx \delta R_{k-1} \hat{R}_{k-1} \exp\{ (w_k)^{\wedge} \} \hat{R}_{k-1}^T \nonumber \\
                     &= \delta R_{k-1} \exp\{ (\hat{R}_{k-1} w_k)^{\wedge} \} 
\end{align}
Employing (\ref{eq:appro_R}) and retaining the first-order term, the linearization of (\ref{eq:right_err_ki}) is computed as
\begin{equation*}
    \xi_{k|k-1} = \xi_{k-1} + \hat{R}_{k-1} w_k,
\end{equation*}
and the covariance matrix of $\xi_{k|k-1}$ is given by
\begin{equation}
    P_{k|k-1} = P_{k-1} + \hat{R}_{k-1} Q_k \hat{R}_{k-1}^T,\label{eq:R_pro_P}
\end{equation}
where $P_{k-1}$ is determined by the approximate relationship with $N_{k-1}$ in Proposition \ref{pro:GaApp}. The propagation \eqref{eq:R_pro_P} can be rewritten with information matrices as
\begin{equation}
    \bar{N}_{k|k-1} = \bar{N}_{k-1} - \bar{N}_{k-1}\left( B_k^{-1} + \bar{N}_{k-1} \right)^{-1}\bar{N}_{k-1},\label{eq:R_pro_N}
\end{equation}
where $B_k = \hat{R}_{k-1} Q_k \hat{R}_{k-1}^T$ and $ B_k^{-1} + \bar{N}_{k-1}\succ 0$. Then, $N_{k|k-1}$ is recovered from  $ \bar{N}_{k|k-1}$ by (\ref{eq:N2P_inv}).

\subsubsection{Likelihood}
The measurement update stage constructs a probability model defined by (\ref{eq:right_MF})  as 
the likelihood distribution, using measurement equations. 

Consider the measurement model described by (\ref{eq:AttMea}). For a given weight matrix $W_k = \text{diag}[w^1_k, \dots, w^n_k]$, the attitude measurement is generated by solving Wahba's problem with  measurement vectors $\tilde{B}_k = \left[\tilde{b}_k^1, \dots, \tilde{b}_k^n \right]$ and reference vectors $E_k = \left[e_k^1, \dots, e_k^n \right]$, which is
\begin{equation}\label{eq:Wahba}
    \hat{R}_k^m(\tilde{B}_k, E_k)\! = \!\mathop{\text{min}}\limits_{\hat{R} \in SO(3)} \frac{1}{2}\!\text{tr}\!\left[\! (E_k\!-\hat{R}^T\!\tilde{B}_k\!)^T\!W_k\!(\!E_k\!-\!\hat{R}^T\!\tilde{B}_k\!) \!\right].
\end{equation}
One solution of Wahba's problem was given by \cite{markleyAttitudeDeterminationUsing1988}, in which the proper SVD was used to compute the optimal attitude.
\begin{lemma} \cite{markleyAttitudeDeterminationUsing1988} \label{le:SVD_at}
     Define a matrix $\tilde{L}_k$ as $\tilde{L}_k=E_k W_k \tilde{B}_k^T$. Then, the unique minimum to (\ref{eq:Wahba}) is obtained as 
    \begin{equation*}
        \hat{R}_k^m = UV^T,
    \end{equation*}
    where $U$and $V$ are given by the proper SVD $\tilde{L}_k = USV^T$. Furthermore, $\hat{R}_k^m$ is unique when $s_2 + s_3 \neq 0$.
\end{lemma}
We assume that the optimal attitude is always unique and its uncertainty is characterized by the conditional probability distribution 
\begin{equation}
     \hat{R}_k^m | R_k  \sim \mathcal{P_R M}(R_k, N_k^m),\label{eq:right_mea}
\end{equation}
where $N_k^m$ is the parameter to be computed. The actual attitude $R_k$ and  vectors $B_k$ in $\mathcal{F_B}$ corresponding to $E_k$ can be rewritten into
\begin{align}
    R_k &= \delta R_k^T \hat{R}_k^m ,{\ }\delta R_k \sim \mathcal{M}(N_k^m), \label{eq:act_att}\\
    B_k &= \tilde{B}_k +\delta B_k,{\ }\delta B_k = [\delta b_k^1, \dots, \delta b_k^n].\label{eq:act_mea}
\end{align}    
According to \cite{sanyal2006globaloptimalattitudeestimation}, the optimal attitude $\hat{R}_k^m$ satisfies
\begin{equation}\label{eq:opti_mea}
    \tilde{L}_k^T \hat{R}_k^m = (\hat{R}_k^m )^T \tilde{L}_k,
\end{equation}
and  $L_k=E_k W_k B_k^T$ satisfies
\begin{equation}\label{eq:opti_real}
    L_k^T R_k = R_k^T L_k.
\end{equation}
Equation (\ref{eq:opti_real}) is linearized with $\xi_k^m =\log(\hat{R}_k^m R_k^T)^{\vee}$ to compute $N_k^m$ as summarized in the following theorem.
\begin{theorem}\label{th:lin_opti}
    Consider a unique $\hat{R} \in SO(3)$ satisfying (\ref{eq:opti_mea}) with measurement vectors $\tilde{B} = \left[\tilde{b}^1, \dots, \tilde{b}^n \right]$ , reference vectors $E = \left[e^1, \dots, e^n \right]$ and a weight matrix $W = {\rm diag}[w^1, \dots, w^n]$. The actual attitude and vectors in $\mathcal{F_B}$ corresponding to $E$ are defined as (\ref{eq:act_att}) and (\ref{eq:act_mea}) respectively. The first-order approximation of (\ref{eq:opti_real}) with errors $\xi =\log( \hat{R}R^T)^{\vee} \in \mathbb{R}^3$ and $\delta B = \left[ \delta b^1,\dots, \delta b^n \right]$ is
\begin{equation}\label{eq:lin_mean}
    \left( {\text{tr}\left( {\tilde{L}\hat{R}^T } \right)I_3  - \tilde{L}\hat{R}^T } \right)\xi  = \sum\limits_i {w^i (e^i )^{\wedge}\hat{R}\delta b^i },
\end{equation}
where $\tilde{L} = EW\tilde{B}^T$.
\end{theorem}
\begin{proof}
    Substituting (\ref{eq:act_att}), (\ref{eq:act_mea}) and (\ref{eq:appro_R}) into  (\ref{eq:opti_real}) and retaining the first-order term, we have
    \begin{equation}
        \hat{R}^T \xi^{\wedge} \tilde{L} + \tilde{L}^T \xi^{\wedge} \hat{R} = \delta B WE^T \hat{R} - \hat{R}^T EW \delta B^T.\label{eq:lin_opt_1}
    \end{equation}
    Multiplying  $\hat{R}$ and $\hat{R}^T$ to the left and right sides of (\ref{eq:lin_opt_1}) yields
    \begin{equation}
         \xi^{\wedge} \tilde{L} \hat{R}^T + \hat{R} \tilde{L}^T \xi^{\wedge}  = 
\hat{R} \delta B WE^T - EW (\hat{R} \delta B)^T.\label{eq:lin_opt_2}
    \end{equation}
    Applying (\ref{eq:crossTr2}) to the left side of (\ref{eq:lin_opt_2}), expanding the right of (\ref{eq:lin_opt_2}) and applying (\ref{eq:crossTr3}), we have
    \begin{align*}
        &\left[ \text{tr}\left( \tilde{L} \hat{R}^T \right)I_3 - \tilde{L} \hat{R}^T \right]\xi\\
        &= \sum_i^n w^i \left[ \left( \hat{R} \delta b^i \right)(e^i)^T - e^i \left( \hat{R} \delta b^i \right)^T \right] = \sum_i^n w_i (e^i)^{\wedge} \hat{R} \delta b^i,
    \end{align*}
    which gives (\ref{eq:lin_mean}).
\end{proof}
Theorem \ref{th:lin_opti} establishes a linear map between $\xi$ and $\delta b_i$, and it is a bijection
 when the optimal attitude given by (\ref{eq:Wahba}) is unique.
\begin{proposition}\label{pro:right_inv}
    Suppose that the optimal attitude given by (\ref{eq:Wahba}) is unique, or equivalently, $s_2 + s_3 \neq 0$. Denoting by $A = {\rm tr}\left( \tilde{L} \hat{R}^T \right)I_3 - \tilde{L} \hat{R}^T$ in (\ref{eq:lin_mean}), then $A$ is invertible.
\end{proposition}
\begin{proof}
    Because $s_1 \geq s_2 \geq |s_3|$ and $s_2 +s_3 \neq 0$, it follows $s_1+s_2 \geq s_1+s_3 \geq s_2+s_3 > 0$. Furthermore, we have
\begin{align*}
    \text{det} \left( A \right) &= \text{det}\left( \text{tr}\left( \tilde{L} \hat{R}^T \right)I_3 - \tilde{L} \hat{R}^T \right) \nonumber\\
                        &= \text{det}\left( \text{tr}\left( USV^T VU^T \right)UI_3 U^T- USV^T VU^T  \right) \nonumber\\
                        &= \text{det} \left( U\ \text{diag}[s_2+s_3, s_1+s_3, s_1+s_2]\ U^T \right)\nonumber\\
                        &= (s_2+s_3)(s_1+s_3)(s_1+s_2)>0.\nonumber
\end{align*}
Therefore, the matrix $A$ is invertible.
\end{proof}
Proposition \ref{pro:right_inv} implies that $A_k$ is invertible and thus equation (\ref{eq:lin_mean}) can be rewritten into
\begin{equation*}
    \xi_k^m = \sum_i^n w^i A_k^{-1} (e_k^i)^{\wedge} \hat{R}^m_k \delta b^i_k,
\end{equation*}
and the covariance matrix of $\xi_k^m$ is given by
\begin{equation*}
    P^m_k = \sum_i^n (w^i)^2 \left(A_k^{-1} (e^i_k)^{\wedge} \hat{R}^m_k\right) G^i_k \left(A_k^{-1} (e^i_k)^{\wedge} \hat{R}^m_k\right)^T.
\end{equation*}
Then $N^m_k$ can be computed with $P^m_k$  by (\ref{eq:N2P}).

\subsubsection{Posterior}
The posterior probability of the attitude is constructed by Bayes' rule with the prior and the likelihood as follows.

\begin{theorem}\label{th:right_post}
    Suppose that the a priori attitude distribution of $R$ is given by $\mathcal{P_R M}(R^-,N^-)$ with a mean attitude $R^- \in SO(3)$ and $ N^- = (N^-)^T \in  \mathbb{R}^{3 \times 3}$. Consider an attitude measurement $ R^m$ satisfying  $ {R}^m | R  \sim \mathcal{P_R M}(R, N^m)$ with $R \in SO(3)$ and  $ N^m = (N^m)^T \in \mathbb{R}^{3 \times 3}$ . Then, the posterior attitude distribution with a measurement $R^m = \tilde{R}^m$ is characterized by,
    \begin{equation}
     R|\tilde{R}^m \sim \mathcal{P_R M}(R^+,N^+),\label{eq:right_post}
    \end{equation}
    where $R^+$ and  $N^+$ are the polar and the elliptical components from the proper left polar decomposition of  $F=N^-R^-+N^m\tilde{R}^m$ respectively.
\end{theorem}
\begin{proof}
    Substituting $p(R)$ and $p(\tilde{R}^m|R)$ into Bayes' rule
    \begin{align*}
         p\left( R|\tilde{R}^m \right) \propto p(R)p(\tilde{R}^m|R),
    \end{align*}
we can obtain
    \begin{align*}
                              p\left( R|\tilde{R}^m  \right) & \propto \text{etr} \left[ \left( N^- R^- \right)^T R+ \left( N^m R \right)^T \tilde{R}^m \right]\nonumber\\
                               &= \text{etr} \left[ \left( N^- R^- + N^m \tilde{R}^m \right)^T R \right], 
    \end{align*}
  which implies that $ R|\tilde{R}^m \sim \mathcal{M}(N^- R^- + N^m R^m)$. Let $F =  N^- R^- + N^m \tilde{R}^m$ and denote the proper left polar decomposition of $F$ by $F=N^+R^+$. According to (\ref{eq:right_MF}), we have $  R|\tilde{R}^m\sim \mathcal{P_R M}(R^+,N^+)$, which gives (\ref{eq:right_post}).
\end{proof}
Because the updating frequency of vector measurements is often lower than that of gyros in practice, the propagation step will be executed in a loop until new measurement vectors or attitudes characterized by $ {R}^m_k | R_k  \sim \mathcal{P_R M}(R_k, N^m_k)$ are available. The pseudo-code for the proposed estimation scheme is presented in Algorithm \ref{alg:right_est}.

\begin{algorithm}
	\caption{Filtering with the right-invariant error}
	\begin{algorithmic}[1] 
		\Procedure{Estimation $R_k$}{} 
  
		\State Let $k=0$,$R_0 \sim \mathcal{P_R M}(\hat{R}_0,N_0)$
		\Loop
            \State $k=k+1$
    		\State $\hat{R}_{k|k-1} = \hat{R}_{k-1} \exp\{ (h\tilde{\omega}_k)^{\wedge} \}$
            \State Compute $\bar N_{k-1}$ with ${N}_{k-1}$ by \eqref{eq:N2P_inv} and $\bar{N}_{k|k-1}$ with $\bar{N}_{k-1}$ and $Q_k$ by ( \ref{eq:R_pro_N}) 
            \State Compute $N_{k|k-1}$ with $\bar{N}_{k|k-1}$ by \eqref{eq:N2P_inv}.
    		\If{ $\tilde{B}_k$ is available}
                \State $\tilde{L}_k=E_k W_k \tilde{B}_k^T$
                \State Compute the proper left polar decomposition of $\tilde{L}_k$ as $\tilde{L}_k=K^m_k M^m_k$ and $\hat{R}_k^m = U_k^m(V_k^m)^T$
                \State $P_k^m = \sum_i^n (w^i)^2 \left(A_k^{-1} (e^i)^{\wedge} \hat{R}_k^m\right) G_k^i \left(A_k^{-1} (e^i)^{\wedge} \hat{R}_k^m\right)^T$
                \State Compute the proper SVD of $P_k^m$ as $P_k^m = U_k^m \Lambda_k^m (U_k^m)^T$
                \State $ N_k^m = U_k^m \left( \frac{1}{2}{\rm tr}\left((\Lambda_k^m)^{-1}\right)I_3 - (\Lambda_k^m)^{-1} \right) (U_k^m)^T $
                \State $F_k = N_k^m \hat{R}_k^m + N_{k|k-1} \hat{R}_{k|k-1}$
            \ElsIf{$R_k^m=\tilde{R}_k^m$ is available}
                \State $F_k = N_k^m \tilde{R}_k^m + N_{k|k-1} \hat{R}_{k|k-1}$
            \Else
                \State $F_k= N_{k|k-1} \hat{R}_{k|k-1}$
            \EndIf
             \State Compute $N_k$ and $\hat{R}_k$ by the proper left polar decomposition of $F_k$ as $F_k =N_k \hat{R}_k$
        \EndLoop
		\EndProcedure
        
	\end{algorithmic}\label{alg:right_est}
\end{algorithm}

\subsection{Filtering with the left-invariant error on SO(3)}
\label{sec:left_filter}
The preceding attitude filter with the left-invariant error is designed based on the uncertainty model defined by (\ref{eq:left_MF}). Similarly,  we can also derive a filter by approximating the underlying posterior distribution of attitudes with (\ref{eq:left_MF}), i.e.,
\begin{align}
    R_k|\{\tilde{b}_k^i\}\sim \mathcal{P_L M}(\hat{R}_k, N_k),\ \delta R = \hat{R}^T R.\label{eq:err_left}
\end{align}
\subsubsection{Prior}
For given $\hat{R}_{k-1}$ and $\tilde{\omega}_k$, the central attitude $\hat{R}_{k|k-1}$ is obtained by the deterministic part of the kinematics, or equivalently, by (\ref{eq:pro_R}). 

Next, $N_{k|k-1}$ is derived from linearized error equations with  redefined $\delta R$ in (\ref{eq:err_left}) as follows. Substituting (\ref{eq:pro_R}) and (\ref{eq:AttitudeKm}) into  (\ref{eq:err_left}), we have
\begin{align*}
    \delta R_{k|k-1} &= \text{exp}\{ (h \tilde{\omega}_k)^{\wedge} \}^T \hat{R}_{k-1}^T R_{k-1} \text{exp}\{ (h \tilde{\omega}_k + w_k)^{\wedge} \} \nonumber \\
                     &= \text{exp}\{ (h \tilde{\omega}_k)^{\wedge} \}^T \delta R_{k-1} \ \text{exp}\{ (h \tilde{\omega}_k + w_k)^{\wedge} \} \nonumber \\
\end{align*}
Assuming the time step $h$ and error $\delta R$ to be small, employing (\ref{eq:appro_R}) and retaining the first-order term, the linearized error equation is derived as
\begin{equation*}
    \xi_{k|k-1} = \exp\{ (h \tilde{\omega}_k)^{\wedge} \}^T\xi_{k-1} + w_{k-1},
\end{equation*}
and the covariance matrix of $\xi_{k|k-1}$ is given by
\begin{equation}\label{eq:L_pro_P}
    P_{k|k-1} = \exp\{ (h \tilde{\omega}_k)^{\wedge} \}^TP_{k-1}\exp\{ (h \tilde{\omega}_k)^{\wedge} \} + Q_k.
\end{equation}
where $P_{k-1}$ is determined by the approximate relationship with $N_{k-1}$ in Proposition \ref{pro:GaApp}. The propagation \eqref{eq:L_pro_P} can be rewritten with information matrices as
\begin{equation}
    \bar{N}_{k|k-1} = \bar{N}'_{k-1} - \bar{N}'_{k-1}\left( Q_k^{-1} + \bar{N}'_{k-1} \right)^{-1}\bar{N}'_{k-1},\label{eq:L_pro_N}
\end{equation}
where $\bar{N}'_{k-1}=\Phi_k^T\bar{N}_{k-1}\Phi_k$, $\Phi_k=\exp\{ (h\tilde{\omega}_k)^{\wedge} \}$, and $Q_k\succ 0$ ensures $Q_k^{-1}+\bar{N}'_{k-1}\succ 0$. Then, $N_{k|k-1}$ is recovered from $\bar{N}_{k|k-1}$ by (\ref{eq:N2P_inv}).

\subsubsection{Likelihood}
For measurement updates, the attitude measurement is obtained uniquely from solving Wahba's problem using SVD method with $\tilde{B}_k$ and $E_k$, and follows
\begin{equation}\label{eq:left_mea}
    \hat{R}_k^m | R_k \sim \mathcal{P_L M}(R_k, N_k^m),
\end{equation}
where $N_k^m$ is computed using (\ref{eq:opti_mea}) and the linearized system equation with the left-invariant error $\delta R_k = R_k^T\hat{R}_k^m$ as follows.
\begin{theorem}
      Consider a unique $\hat{R} \in SO(3)$ satisfying (\ref{eq:opti_mea}) with measurement vectors $\tilde{B} = \left[\tilde{b}^1, \dots, \tilde{b}^n \right]$ , reference vectors $E = \left[e^1, \dots, e^n \right]$ and a weight matrix $W = {\rm diag}[w^1, \dots, w^n]$. The first-order approximation of (\ref{eq:opti_real}) with $\delta R = R^T\hat R$ and $\delta B = \left[ \delta b^1,\dots, \delta b^n \right]$ satisfying (\ref{eq:act_mea}) is
\begin{equation}\label{eq:lin_mean_l}
    \left( {\text{tr}\left( {\hat{R}^T \tilde{L} } \right)I_3  - \hat{R}^T\tilde{L} } \right)\xi  = \sum\limits_i {w^i (\hat{R}^T e^i )^{\wedge}\delta b^i },
\end{equation}
where $\xi =\log( \delta R)^{\vee} \in \mathbb{R}^3$ and $\tilde{L} = EW\tilde{B}^T$.
\end{theorem}
\begin{proof}
    Substituting $\hat R_k^m=R_k\delta R_k$ (i.e., $\delta R_k=R_k^T\hat R_k^m$) and (\ref{eq:appro_R})into (\ref{eq:opti_real}) and retaining the first-order term, we have 
    \begin{equation}\label{eq:lin_opt_left_1}
         \xi^{\wedge} \hat{R}^T \tilde{L} + \tilde{L}^T \hat{R} \xi^{\wedge} = \delta B WE^T \hat{R} - \hat{R}^T EW \delta B^T.
    \end{equation}
    Applying (\ref{eq:crossTr2}) to the left side of (\ref{eq:lin_opt_left_1}), expanding the right side of the  (\ref{eq:lin_opt_left_1}) and applying (\ref{eq:crossTr3}), we have
    \begin{align*}
        \left[ \text{tr}\left( \hat{R}^T \tilde{L} \right)I_3 - 
\hat{R}^T \tilde{L} \right]\xi
        &= \sum_i^n w^i \left[  \delta b^i (\hat{R}^T e^i)^T - \hat{R}^T e^i (\delta b^i)^T \right] \nonumber\\
        &= \sum_i^n w^i (\hat{R}^Te^i)^{\wedge} \delta b_i,
    \end{align*}
    which gives (\ref{eq:lin_mean_l}).
\end{proof}
Next, we prove the invertibility of the matrix ${\rm tr}\left( \hat{R}^T \tilde{L} \right)I_3 - 
\hat{R}^T \tilde{L}$ such that a linear bijection can be established.
\begin{proposition}
    Suppose that the optimal attitude given by (\ref{eq:Wahba}) is unique, or equivalently, $s_2 + s_3 \neq 0$. Denoting $B= {\rm tr}\left( \hat{R}^T \tilde{L} \right)I_3 - 
\hat{R}^T \tilde{L}$ in (\ref{eq:lin_mean_l}), then $B$ is invertible.
\end{proposition}
\begin{proof}
    Because $s_1 \geq s_2 \geq |s_3|$ and $s_2 +s_3 \neq 0$, $s_1+s_2 \geq s_1+s_3 \geq s_2+s_3 > 0$. Furthermore, we have
\begin{align*}
    \text{det}\left(B \right) &= \text{det}\left( {\rm tr}\left( \hat{R}^T \tilde{L} \right)I_3 - 
\hat{R}^T \tilde{L} \right) \nonumber\\
                        &= \text{det} \left( V\ {\rm tr}[s_2+s_3, s_1+s_3, s_1+s_2]\ V^T \right)\nonumber\\
                        &= (s_2+s_3)(s_1+s_3)(s_1+s_2)>0.\nonumber
\end{align*}
Therefore, the matrix $B$ is invertible.
\end{proof}
Using (\ref{eq:lin_mean_l}), the covariance matrix of $\xi^m_k$ is given by
\begin{equation}
    P^m_k\! =\! \sum_i^n (w^i)^2 \!\left(\!A_k^{-1} \left((\hat{R}_k^m)^Te_k^i\right)^{\wedge}\!\right) \!G^i_k \!\left(\!A_k^{-1} \left((\hat{R}_k^m)^Te_k^i\!\right)^{\wedge}\!\right)^T.\label{eq:pm_L}
\end{equation}
\subsubsection{Posterior}
Similarly to the proposed filter with the right-invariant error, the following theorem is presented to compute the parameters of the posterior distribution.
\begin{theorem}
    Suppose that the a priori attitude distribution of $R$ is given by $\mathcal{P_L M}(R^-,N^-)$ with a central attitude $R^- \in SO(3)$ and $   N^- \in \mathbb{R}^{3 \times 3}$. Consider an attitude measurement $ R^m$ satisfying  $ {R}^m | R \sim \mathcal{P_L M}(R, N^m)$ with $N^m = (N^m)^T \in \mathbb{R}^{3 \times 3}$. Then, the posterior attitude distribution with a measurement $R^m = \tilde{R}^m$ is characterized by
    \begin{equation}
     R|\tilde{R}^m \sim \mathcal{P_L M}(R^+,N^+),\label{eq:left_post}
    \end{equation}
    where $R^+$ and  $N^+$ are the polar and the elliptical components from the proper right polar decomposition of  $F = R^-N^-+\tilde{R}^mN^m$ respectively.
\end{theorem}
\begin{proof}
    Substituting $p(R)$ and $p(\tilde{R}^m|R)$ into Bayes' rule and similarly to the proof of Theorem \ref{th:right_post}, we can obtain
    \begin{align*}
        p\left( R|\tilde{R}^m  \right)  & \propto \text{etr} \left[ \left( 
R^- N^- \right)^T R+ \left( R N^m \right)^T\tilde{R}^m\right]\nonumber\\
                               &= \text{etr} \left[ \left( R^- N^-  +\tilde{R}^m N^m  \right)^T R \right], 
    \end{align*}
  which implies that $R|\tilde{R}^m \sim \mathcal{M}(R^- N^-  + \tilde{R}^m N^m)$. Let $F =  R^- N^-  + \tilde{R}^m N^m$ and denote the right polar decomposition of $F$ by $F=R^+ N^+$. According to (\ref{eq:left_MF}), we have $R|R^m  \sim \mathcal{P_L M}(R^+, N^+)$, which gives (\ref{eq:left_post}).
\end{proof}

 The pseudo code for the filter with the left-invariant error is summarized in Algorithm \ref{alg:left_est}.

\begin{algorithm}
	\caption{Filter with the left-invariant error}
	\begin{algorithmic}[1] 
		\Procedure{Estimation $R_k$}{} 
  
		\State Let $k=0$,$R_0 \sim \mathcal{P_LM}(\hat{R}_0,N_0)$
		\Loop
            \State $k=k+1$
    		\State $\hat{R}_{k|k-1} = \hat{R}_{k-1} \exp\{ (h\tilde{\omega}_k)^{\wedge} \}$
            \State Compute $\bar{N}_{k-1}$ with $N_{k-1}$ by \eqref{eq:N2P_inv} and $\bar N_{k|k-1}$ with $\bar N_{k-1}$ and $Q_k$ from ( \ref{eq:L_pro_N}) 
            \State Compute the $N_{k|k-1}$ with $\bar N_{k|k-1}$ by \eqref{eq:N2P_inv}
    		\If{$\tilde{B}_k$ is available}
                \State $\tilde{L}_k=E_k W_k \tilde{B}_k^T$
                \State Compute the proper left polar decomposition of $\tilde{L}_k$ as $\tilde{L}_k=K^m_k M^m_k$ and $\hat{R}_k^m = U_k^m(V_k^m)^T$
                \State Compute $P_k^m$ by \eqref{eq:pm_L}
                \State Compute the proper SVD of $P_k^m$ as $P_k^m = V_k^m \Lambda_k^m (V_k^m)^T$
                \State $ N_k^m = V_k^m \left( \frac{1}{2}{\rm tr}\left((\Lambda_k^m)^{-1}\right)I_3 - (\Lambda_k^m)^{-1} \right) (V_k^m)^T $
    \Statex
                \State $F_k = \hat{R}_k^m N_k^m  + \hat{R}_{k|k-1} N_{k|k-1} $
            \ElsIf{$R^m_k=\tilde{R}^m_k$ is available}
                \State $F_k = \tilde{R}_k^m N_k^m + \hat{R}_{k|k-1} N_{k|k-1}$
            \Else
                \State ${F}_k = {F}_{k|k-1}$ 
            \EndIf
            \State Compute $N_k$ and $\hat{R}_k$ by the proper right polar decomposition of $F_k$ as $F_k = \hat{R}_k N_k$
        \EndLoop
		\EndProcedure
        
	\end{algorithmic}\label{alg:left_est}
\end{algorithm}

\subsection{Discussions}
The two proposed estimation methods are partially similar to the IEKF in \cite{barrauInvariantKalmanFiltering2018} and the Bayesian estimator in \cite{leeBayesianAttitudeEstimation2018} but differ greatly in some critical steps. More precisely, our methods utilize linearized error equations and two newly defined probability models to approximate the underlying prior and likelihood distributions on $SO(3)$. As a result, they are closed-form and avoid solving nonlinear equations to obtain parameters from the first moment, significantly reducing computational complexity compared to the MFD-based filter in \cite{leeBayesianAttitudeEstimation2018} that requires moment matching or unscented transformation. 

The computational complexity of the proposed filter in Algorithm \ref{alg:right_est} is contributed by the proper SVD, matrix exponential function, addition, and multiplication. Since the proper SVD for $3\times3$ symmetric matrix and the matrix exponential function for $3\times3$ anti-symmetric matrix admit closed-form solutions, their computational costs are constant, i.e. $O(1)$. Therefore, the overall per-step complexity is mainly dominated by the computation of $\tilde{L}_k$, which scales linearly as $O(n)$, where $n$ is the number of reference vectors. For the first-order estimator, however, additional complexity is caused by solving
\begin{equation}\label{eq:moment_pro}
    \frac{1}{c(S)}\frac{\partial c(S)}{\partial s_i} -d_i =0,\ for\ i=\{1,2,3\},
\end{equation}
with the Newton-Armijo iteration, as presented in \cite{Lee2017matrix}, in moment matching. Each iteration needs to compute $c(S),\ \partial c(S)/\partial s_i,\ \partial^2 c(S)/\partial s_i\partial s_j$ by the integration of modified Bessel functions, whose complexity scales linearly with the number of quadrature points, denoted by $n_q$. If the average number of trials to satisfy the Armijo line search is denoted by $n_L$, and assuming that the solver converges within $n_I$ iterations, the per-step computational complexity of the first-order estimator is $O(n_I(1+n_L)n_q)+O(n)$, where the term $O(1)$ represents the cost of computing the descent direction. Thus, the extra computational cost of the first-order estimator is $O(n_I(1+n_L)n_q)$. It is notable that there is \textit{no} theory guarantees that the Newton-Armijo iteration terminates in a finite number of steps and $n_I$, as a result, can be very large. Even if $n_I$ is prespecified as a $n_I^{max}$, such as $n^{max}_I =100$ in \cite{Lee2017matrix}, $O(n_I(1+n_L)n_q)$ is still larger than $O(n)$ because $n_q$ is usually set to $50\sim 100$ to promise accuracy, but $n$ can be a single digit.


It should be acknowledged that linearization, while improving computational speed, results in a theoretical loss of accuracy. For the estimator in \cite{leeBayesianAttitudeEstimation2018}, $\mathcal{O}(h^{1.5})$ is omitted in moment propagation. In error equations (\ref{eq:right_err_ki}) and (\ref{eq:lin_mean}), instead, high order terms $\mathcal{O}(\|\xi\|^2)$ and $\mathcal{O}(h^2)$ are omitted. Since $\xi$ represents estimation error, $\mathcal{O}(\|\xi\|^2)$ can be non-negligible when the uncertainty is large, even if $h$ is small. However, this does not imply that the estimator in \cite{leeBayesianAttitudeEstimation2018}
is always more accurate than the proposed filter. The numerical solution of (\ref{eq:moment_pro}) cannot guarantee convergence within the specified $n_I^{max}$, leading to accuracy loss that cannot be precisely computed. The two estimators with MFDs are further compared with simulations in Section \ref{sec: Simulation}.

While linearized equations are used to propagate uncertainty along the kinematics, the proposed filters fuse measurements and extract the posterior estimates in a global fashion. Extracting estimates of the posterior MFDs is transformed into solving a classic Wahba problem \cite{Wahba1965ALS} by the SVD method, where the weight matrices of the cost are parameters $N$ of the prior and likelihood MFDs. In contrast, the Bayesian filters with CGDs do not naturally admit a similar global optimization formulation since local coordinates are used to represent uncertainty. Moreover, due to the introduction of MFDs, the weight matrices establish a quantitative relationship with the process and measurement noises, thereby endowing the solution to the Wahba problem with probabilistic optimality. Instead, Kalman filters, such as the IEKF, perform only a local Gauss--Newton step in the correction with measurements. The effect of the weight matrices given by MFDs on performance is demonstrated by ablation experiments in the next section.

\section{The proposed Filter as a Stable Observer} \label{sec:stability}

\subsection{Stability analysis on $SO(3)$}
The aim of this subsection is to prove the stability of the proposed filter around any trajectory on $SO(3)$ as a deterministic observer. The stability analysis considers the observer form of the filter, in which noise-free measurements are used. The measurement noise-related parameters $N^m_k$ and $Q_k$ are  design gains that can be adjusted \textit{a priori}. The following theorem proves the almost global asymptotic stability of this observer form.
\begin{theorem}\label{th:global_stb}
    Consider the discrete attitude kinematics $R_k=R_{k-1}\exp(h\omega_k^{\wedge})$ with the noise-free angular velocity measurement $\tilde{w}_k$ and exact attitude measurement $\hat{R}_k^m$. Let $ \alpha_1I_3 \preceq Q_k \preceq \alpha_2I_3$ and $\beta_1I_3\preceq N^m_k=n^m_kI_3\preceq\beta_2I_3$ for arbitrary $k$, where $\alpha_2\geq\alpha_1>0$ and $\beta_2\geq\beta_1>0$. Then the observer form of  the proposed filter is almost globally asymptotically stable as an observer, i.e., it is asymptotically stable when the initial attitude satisfies ${\rm tr}(I_3-R_0\hat{R}_0^T)\leq4-\epsilon$ for any $\epsilon \in (0,4]$.
\end{theorem}
\begin{proof}
    Consider a candidate Lyapunov function as
    \begin{align*}
        V_k = {\rm tr} \left( I_3 - R_k\hat R_k^T \right) ,
    \end{align*}
    and define $ V_{k|k-1} = {\rm tr} \left( I_3 - R_k\hat R_{k|k-1}^T \right) $. Since angular velocity measurements are noise-free, there is ${R_{k - 1}}\hat R_{k - 1}^T = {R_k}\hat R_{k|k - 1}^T$, and thus $V_{k-1} =V_{k|k-1}$.

    Since $\hat{R}_k$ is the mean attitude of $\mathcal{M}(F_k)$, we have ${{\hat R}_k} = \mathop {{\rm{max}}}\limits_{R \in SO(3)} {\rm{tr}}(F_k^TR)$, as presented in Theorem 2.3 in \cite{leeBayesianAttitudeEstimation2018}. Substituting $F_k = N_k^m\hat{R}_k^m+ N_{k|k-1}\hat{R}_{k|k-1}$, the following inequality is obtained as
    \begin{align*}
       \begin{array}{l}
{\rm{tr}}\left[ {{{\left( {{{\hat R}_{k|k - 1}}} \right)}^T}{N_{k|k - 1}}{{\hat R}_k}} \right] + {\rm{tr}}\left( {{{\hat R}_k}^T{N^m_k}\hat R_k^m} \right)\\
 \ge {\rm{tr}}\left[ {{{\left( {{{\hat R}_{k|k - 1}}} \right)}^T}{N_{k|k - 1}}{{\hat R}_{k|k - 1}}} \right] + {\rm{tr}}\left( {{{\hat R}_{k|k - 1}}^T{N^m_k}\hat R_k^m} \right)\\
 = {\rm{tr}}\left[ {{N_{k|k - 1}}} \right] + {\rm{tr}}\left( {{N^m_k}\hat R_k^m} {{\hat R}_{k|k - 1}}^T\right),
\end{array}
    \end{align*}
    which can be rewritten as
    \begin{align*}
      \begin{array}{l}
{\rm{tr}}\left( {{N^m_k}\hat R_k^m{{\hat R}_k}^T} \right) - {\rm{tr}}\left( {{N^m_k}\hat R_k^m{{\hat R}_{k|k - 1}}^T} \right)\\
 \ge {\rm{tr}}\left[ {{N_{k|k - 1}}} \right] - {\rm{tr}}\left[ {{N_{k|k - 1}}{{\hat R}_k}}{{\left( {{{\hat R}_{k|k - 1}}} \right)}^T} \right],
\end{array}
    \end{align*}
    where the equality holds when $\hat{R}_k=\hat{R}_{k|k-1}$. Using Lemma 2.2.6 in\cite{berkane2017hybridlong}, there is
    \begin{align*}
        &{\rm{tr}}\left[ {{N_{k|k - 1}}} \right] - {\rm{tr}}\left[ {{N_{k|k - 1}}{{\hat R}_k}}{{\left( {{{\hat R}_{k|k - 1}}} \right)}^T} \right]\geq \\
         &\frac{1}{2}\lambda_{\min}^{\bar{N}_{k|k-1}} {\rm tr}\left[I_3-{{\hat R}_k}{{\left( {{{\hat R}_{k|k - 1}}} \right)}^T} \right],
    \end{align*}
where $\bar{N}_{k|k-1}:= {\rm tr}(N_{k|k-1})I_3-{N}_{k|k-1} = P_{k|k-1}^{-1}$, and $\lambda_{\min}^{\bar{N}_{k|k-1}}$ is its minimum eigenvalue. Denote the proper SVD of $N_{k-1}$ as $N_{k-1}=U_{k-1}S_{k-1}U_{k-1}^T$ and $S_{k-1}={\rm diag}[s_{k-1,1},s_{k-1,2},s_{k-1,3}]$, which satisfies $s_{k-1,1}\geq s_{k-1,2}\geq| s_{k-1,3}|\geq 0$. Thus, $P_{k-1}$ given by Proposition \ref{pro:GaApp} is positive definite. Since $Q_k\succeq \alpha_1I_3 $ is a positive definite matrix, $P_{k|k-1}^{-1}=(P_{k-1}+\hat{R}_{k-1}Q_k\hat{R}_{k-1}^T)^{-1}$ is also positive definite. As a result, we have
${\rm{tr}}\left[ {{N_{k|k - 1}}} \right] - {\rm{tr}}\left[ {{N_{k|k - 1}}{{\hat R}_k}}{{\left( {{{\hat R}_{k|k - 1}}} \right)}^T} \right]\geq 0$, where the equality also holds when $\hat{R}_k=\hat{R}_{k|k-1}$. Thus, there is
\begin{align*}
    &V_k- V_{k-1} =  V_k- V_{k|k-1}\\
    &=\frac{1}{n^m_k}
\left[ {\rm{tr}}\left( {{N^m_k}\hat R_k^m{{\hat R}_{k|k - 1}}^T} \right) - {\rm{tr}}\left( {{N^m_k}\hat R_k^m{{\hat R}_k}^T} \right) \right]\\
 &\le \frac{1}{n^m_k}\left\{ {\rm{tr}}\left[ {{N_{k|k - 1}}{{\hat R}_k}}{{\left( {{{\hat R}_{k|k - 1}}} \right)}^T} \right] - {\rm{tr}}\left[ {{N_{k|k - 1}}} \right] \right\} \le 0,
\end{align*}
which indicates that $V_k$ is monotonically decreasing as $k$ increases, and $\hat{R}_k=\hat{R}_{k|k-1}$ gives equilibrium points. Substituting $\hat{R}_k=\hat{R}_{k|k-1}$ into $F_k=N^m_k\hat{R}_k^m+N_{k|k-1}\hat{R}_{k|k-1}$, we have
\begin{align*}
  {N_k}{\hat R_{k}} = {N_{k|k - 1}}{\hat R_{k}} + {N^m_k}\hat R_k^m,
\end{align*}
which can be rewritten as ${N_k} = {N_{k|k - 1}} + {N^m_k}\hat R_k^m {\hat R_{k}}^T$. Since both $N_k$ and $N_{k|k-1}$ are symmetric, ${N^m_k}\hat R_k^m {\hat R_{k}}^T$ is also symmetric. Moreover, because $N_k^m=n^m_kI_3$ and measurements are noise-free, $\hat R_k^m {\hat R_{k}}^T$ is a symmetric matrix belonging to $SO(3)$ and is equal to $R_k {\hat R_{k}}^T$, indicating that $V_k\in\{0,4\}$. Since $V_k$ is monotonically decreasing and $V_k\leq V_0\leq4-\epsilon<4$, we have $V_k=0$ and thus $\hat{R}_k=R_k$ is the unique equilibrium point. Therefore, the filter is asymptotically stable for every initial attitude satisfying ${\rm tr}(I_3-R_0\hat{R}_0^T)\leq4-\epsilon$, and the excluded set ${\rm tr}(I_3-R_0\hat{R}_0^T)=4$, i.e., a $180^\circ$ initial error, is of measure zero. As a result, the filter is almost globally asymptotically stable.
\end{proof}
The theorem demonstrates that the proposed filter has an almost global region of attraction, while the IEKF, as presented in \cite{barrau2016invariant}, possesses local asymptotic stability. The local linear approximations are introduced in the IEKF to compute Kalman gains in the correction step using measurements, which results in a local region of attraction. In contrast, the fusion in the proposed filter is a Wahba-type global update, which, together with the right-invariant construction of the filter, enables the global asymptotic stability of the observer form. Moreover, none of the existing attitude estimators with directional statistics, such as those in \cite{leeBayesianAttitudeEstimation2018}, provide theoretical stability guarantees because they rely on numerical solutions, which lack proven finite-step convergence. Subsequently, we focus on an example of single-axis rotations to further reveal how the nonlinear mechanism presented in Section \ref{sec:filter_mec} influences the convergence rate of the filter with MFDs.

\subsection{Convergence rate analysis} \label{sec:single_axis}
We consider a rigid body rotating about $b_3$ with an initial attitude $R_0=\exp({\theta}_0 b_3^{\wedge})$, i.e., $R_k \in \mathcal{S}_{I_3}(b_3)$. Other single-axis rotational motions can be transformed to be about $b_3$ through coordinate transformation.  Then, the filter can be rewritten in one-dimensional form on $\mathcal{S}_{I_3}(b_3)$ using the following theorem.

\begin{theorem}\label{th:uniaxial}
      Consider the discrete attitude kinematics $R_k=R_{k-1}\exp(h w_k b_3^{\wedge})$ for $w_k\in\mathbb{R}$, the noise-free angular velocity measurement $\tilde{w}_k=w_k$, and the exact attitude measurement $\hat{R}_k^m=R_k$. Denote the true and measured attitudes by $R_k=\exp(\theta_k b_3^{\wedge})$ and $\hat{R}_k^m=\exp(\theta_k^m b_3^{\wedge})$, respectively. Let the design parameters be $Q_k={\rm diag}[\sigma_k^*,\sigma_k^*,\sigma_k]$ and $N_k^m={\rm diag}[\kappa_k^m,\kappa_k^m,\kappa_k^{m*}]$ with $\sigma_k,\sigma_k^*,\kappa_k^m,\kappa_k^{m*}>0$. Given the initial parameters $\hat{R}_0=\exp(\hat{\theta}_0 b_3^{\wedge})$ and $N_0={\rm diag}[\kappa_0,\kappa_0,\kappa_0^{*}]$ with $\kappa_0,\kappa_0^{*}>0$, the filter in Algorithm \ref{alg:right_est} can be rewritten as
    \begin{align}
        &\hat{R}_{k|k-1} = \exp\left( \hat{\theta}_{k|k-1} b_3^{\wedge} \right),\\
        &N_{k|k-1}={\rm diag}[\kappa_{k|k-1},\kappa_{k|k-1},\kappa_{k|k-1}^{*}],\\
        &\hat{R}_{k} =\exp\left( \hat{\theta}_{k} b_3^{\wedge} \right),\quad N_{k}={\rm diag}[\kappa_{k},\kappa_{k},\kappa_{k}^{*}],
    \end{align}
    where
    \begin{align}
        &\hat{\theta}_{k|k-1} = \hat{\theta}_{k-1} + \tilde{w}_k h,\qquad
        \kappa_{k|k-1}=\frac{\kappa_{k-1}}{1+2\kappa_{k-1}\sigma_k},\\
        &\kappa_{k} = \sqrt{(\kappa_{k|k-1})^2+(\kappa_k^m)^2+2\kappa_{k|k-1}\kappa_k^m\cos\left(\theta_k^m-\hat{\theta}_{k|k-1}\right)},\\
     & \hat{\theta}_{k}  = \hat{\theta} _{k|k-1}+\nonumber  \\
     & {\rm{sign}}({\theta}_k^m-\hat{\theta}_{k|k-1})\arccos \left[({{{\kappa  _{k|k-1}  + \kappa _k^m\cos ({\theta}_k^m-\hat{\theta}_{k|k-1} )}}/{{\kappa  _{k} }}} \right],
    \end{align}
    and
    \begin{align}
        &\kappa_{k|k-1}^{*}=\frac{\kappa_{k-1}+\kappa_{k-1}^{*}}{1+\sigma_k^{*}\left(\kappa_{k-1}+\kappa_{k-1}^{*}\right)}-\frac{\kappa_{k-1}}{1+2\kappa_{k-1}\sigma_k},\\
        &\kappa_{k|k}^{*}=\kappa_{k|k-1}^{*}+\kappa_k^{m*}.
    \end{align}
\end{theorem}
\begin{proof}
    See Appendix A.
\end{proof}
For single-axis rotational motion, Theorem \ref{th:uniaxial} provides a one-dimensional parameter representation of the filter with $\kappa_{k|k-1},\hat{\theta}_{k|k-1},\kappa_k,\hat{\theta}_k$ on $\mathcal{S}_{I_3}(b_3)$, and indicates that the attitude estimates are independent of $\kappa_{k|k-1}^*,\kappa_k^*,\kappa_k^{m*}$ and $\sigma_k^*$. Moreover, the posterior parameters $\kappa_k$ and $\hat{\theta}_k$ are the same as the posterior mean angle $\bar{\theta}^+$ and concentration parameter $\kappa^+$, implying that the nonlinear filtering mechanism discussed in Section \ref{sec:filter_mec} is faithfully reflected in the filter for single-axis rotation. Next, we focus on the one-dimensional filter on $\mathcal{S}_{I_3}(b_3)$ and show how the filtering mechanism guarantees the local exponential stability.

For single-axis rotation, the error defined by the Frobenius norm can be rewritten as $\| R_k -\hat{R}_k \|_F^2 =4[1-\cos(\theta_k-\hat{\theta}_k)] \in [0,8]$. Combining the concentration parameter and estimation error, define a candidate Lyapunov function by 
\begin{equation}\label{eq:de_ly}
    V_k= \kappa_k\left[1-\cos(\theta_k-\hat{\theta}_k)\right].
\end{equation}
The following theorem proves local exponential stability with $V_k$.
\begin{theorem}\label{th:cov}
Suppose that there exist $\alpha_1,\alpha_2,\beta_1,\beta_2>0$ such that $\alpha_1  <\sigma_k<\alpha_2$ and $\beta_1<\kappa_k^m<\beta_2$,  $\forall k \geq 1$. Then the one-dimensional filter in Theorem \ref{th:uniaxial} is  locally exponentially stable. More specifically, $\forall \hat{R}_0 \in \mathcal{S}_{I_3}(b_3)$ such that ${\rm tr}(I_3-R_0\hat{R}_0^T)<2$, there is 
    \begin{equation}\label{eq:cov_speed}
     {\| R_k -\hat{R}_k\|_F^2}\leq C \lambda^k\| R_0 -\hat{R}_0 \|_F^2,
    \end{equation}
where $\lambda= \left( 1+2\alpha_1\beta_1\right)^{-1}$ and independent of $k$ and $C={\kappa_0}/({\beta_1\lambda})$ is a constant.
\end{theorem}
\begin{proof}
    See Appendix B. 
\end{proof}
Theorem \ref{th:uniaxial} and \ref{th:cov} combine the analysis in Section \ref{sec:filter_mec} with the convergence rate of the proposed filter with MFDs as a one-dimensional observer. Together with Theorem \ref{th:global_stb}, the results in this section provide further theoretical support for the stability of the proposed filters in the presence of large initial errors as an observer.

\section{Numerical Simulations}
\label{sec: Simulation}
This section compares the proposed fast nonlinear filters, the MFD-based filters in \cite{leeBayesianAttitudeEstimation2018} and \cite{leeBayesianAttitudeEstimation2018b}, and the IEKF \cite{barrauInvariantKalmanFiltering2018}. For simplicity, we abbreviate our fast nonlinear filters with MFDs derived from the right- and left-invariant error by FNF-R and FNF-L, the Bayesian filter with MFDs in \cite{leeBayesianAttitudeEstimation2018} by BF-MFD, the Bayesian filter with approximated MFDs by BF-AMFD \cite{leeBayesianAttitudeEstimation2018b}.  A filter, denoted by FNF-A, is designed to be the same as the FNF-R, except that the weight matrices $N_{k|k-1}$ and $N_k^m$, are no longer derived from the propagation and approximation with MFDs but are instead set to identity matrices. FNF-A is compared with other filters as ablation experiments to demonstrate how much of the performance of the proposed filters comes from the optimal Wahba attitude solution rather than from the introduction of MFDs. Similarly to \cite{leeBayesianAttitudeEstimation2018}, a 3D pendulum is adopted to generate true attitude and angular velocity for numerical simulations. The initial attitude and angular velocity are given as $R_0 = I_{3\times 3}$ and $\Omega_0= 4.14 \times [1,1,1]^T rad/s$ respectively. The attitude is assumed to be estimated with angular velocity and vector measurements at 50 Hz and 10 Hz respectively. The white noise of angular velocity measurements is Gaussian with $H = \sigma I_{3 \times 3}$ and $\sigma = 1\  deg/\sqrt{s}$. The vector measurement errors follow the Gaussian distribution with zero mean and a covariance that varies by each example. For all examples, the simulation time is $T = 60s$ and the time step is $h = 0.02s$. The SVD method is utilized to estimate the attitude directly from vector measurements and calculate the measurement error. For each case, 50 Monte Carlo simulations (with respect to the random noise) are carried out. All simulations are conducted on a computer owning a 6-core Intel i9-9900 3.10 GHz CPU, and 32G RAM.

Note that the two MFD-based Bayesian filters in \cite{leeBayesianAttitudeEstimation2018} and \cite{leeBayesianAttitudeEstimation2018b} cannot directly deal with common vector measurement errors characterized by Gaussian distributions. To compare the four filters, we use a method presented in \cite{loveGaussianStatisticsPalaeomagnetic2003} to approximate the directional distribution of vector measurements with Gaussian noise. For a random vector $x\in \mathbb{R}^3$ characterized by $x\sim \mathcal{N}(\mu, k^2I_3)$, the possibility distribution density of  $x_1=x/{\Vert x \Vert}$ is given by
\begin{equation}
    p(x_1) = \frac{\kappa}{4\pi \sinh \kappa} \exp(\kappa \mu_1 x_1),
\end{equation}
where $\kappa = \Vert \mu \Vert^2 / k^2$ and $\mu_1 = \mu / \Vert \mu \Vert$. For a more general case with $x\sim  \mathcal{N}(\mu,\Sigma)$, $\kappa$ is given by $3\Vert \mu \Vert^2/\text{tr}(\Sigma)$ to approximate the average dispersion characteristic of the directional distribution of $x$.

\subsection{Small initial errors}\label{sec:small_ie}
In this subsection, measurement noises are chosen to be isotropic with $G_k^i = \sigma_m^2 I_3$ and $\sigma_m \in \mathbb{R}$. Two different levels of covariance are tested: (i) $\sigma_m = 0.24$; (ii) $\sigma_m = 0.04$. The initial parameter of the four filters with MFDs is $F_0 = 10 \exp\left( {\pi e_1^{\wedge}}/{18} \right),$ implying that the initial attitude information has a small error and high confidence. For the IEKF method, the corresponding initial attitude is $R_0 = \exp\left( {\pi e_1^{\wedge}}/{18} \right)$ with a covariance matrix $P_0 = \text{diag}[1/(s_2+s_3), 1/(s_1+s_3),1/(s_1+s_2)] = 0.05 I_3$, where $s_i$ for $i\in\{1,2,3\}$ is the proper singular value of $F_0$. The results are provided in Table \ref{tab:1} for comparison. The CPU time of one simulation is also averaged across all simulations. The attitude error at the $k$-th step of the $i$-th simulation is defined as $E^i_k = \Vert \log(\hat{R}_k^i (R_k)^T) ^{\vee}\Vert,$ where $\hat{R}^i_k$ is the attitude estimate at the $k$-th step in the $i$-th simulation and $R_k$ is the true attitude at the $k$-th step. The average attitude error and the standard deviation are defined, respectively, as
\begin{equation*}
    AE = \frac{h}{MT} \sum_{i=1}^{M} \sum_{k=1}^{T/h} E^i_k,
\end{equation*}
\begin{equation*}
    SD = \sqrt{\frac{1}{M-1} \sum_{i=1}^{M} (\frac{h}{T}\sum_{k=1}^{T/h} E^i_k - AE)^2}.
\end{equation*}

\begin{table*}[htbp]
\centering
\caption{AE (±SD) and time consumption: small initial errors}

\label{tab:1}
\begin{tabular}{c|c|cccccc}
\hline
Measurement Model     &             & FNF-R& FNF-L &FNF-A& BF-MFD&BF-AMFD& IEKF         \\ \hline
{(i)}& est. err(°) & 5.6586±0.4980& 5.6722±0.5075&20.2196±0.5842& 5.6247±0.5106&5.5945±0.4943& 5.7272±0.5541\\
                      & CPU time(s) & 0.4993±0.0441& 0.8141±0.0759&-& 15.4467±1.0391&5.1117±0.4598& 0.3239±0.0498\\ \hline
{(ii)}& est. err(°) & 3.8480±0.2586& 3.8482±0.2597&7.6259±0.1811& 3.8432±0.2474&3.8387±0.2456& 3.8427±0.2552\\
                      & CPU time(s) & 0.6823±0.3429& 1.1648±0.6515&-& 65.3059±5.0346&5.0466±0.5077& 0.3577±0.0933\\ \hline
\end{tabular}
\end{table*}

The results show that the accuracy and convergence rates of these six filters are almost identical. These results corroborate the analyses presented in Section \ref{sec:filter_mec} for a small initial error, in which the differences in $\kappa^+_{MFD}$ and $\kappa^+_{CGD}$, and in $\bar{\theta}^+_{MFD}$ and $\bar{\theta}^+_{CGD}$ are minor since the mean attitudes of the prior and the likelihood distributions are almost identical. However, it is noteworthy that the average time consumption of the two proposed filters is about one to two orders of magnitude less than that of BF-MFD, and about one order of magnitude less than that of BF-AMFD. Moreover, the AE of FNF-A indicates that the ablation of weight matrices calculated by MFDs significantly affects the accuracy of the estimates.

\subsection{Large initial errors}
This subsection considers three cases of measurement noises as (i) $G_k^i = (0.24)^2I_3,\  F_0 = 10^{-3}\exp \left( \pi e_1^{\wedge} \right)$; (ii) $G_k^i = 0.04^2I_3,  F_0 = \exp \left( \pi e_1^{\wedge} \right)$; (iii) $G_k^i = \text{diag}[0.3^2\ 0.01^2\ 0.01^2], F_0 = \exp \left( \pi e_1^{\wedge} \right)$. The first two cases are chosen to be isotropic while the third case is chosen to be non-isotropic. The parameter $F_0 = \exp \left( \pi e_1^{\wedge} \right)$ implies that the initial estimate rotates 180° around $e_1$ compared to the true initial attitude and the initial uncertainties are large. The corresponding initial parameters are set as $R_0 = \exp \left( \pi e_1^{\wedge} \right)$ and $P_0 = 0.5 I_3$. The results are provided in Table \ref{tab:2} and Fig. \ref{fig:large_error}.
\begin{figure}[htbp]
	\begin{center}
		\includegraphics[width=\linewidth]{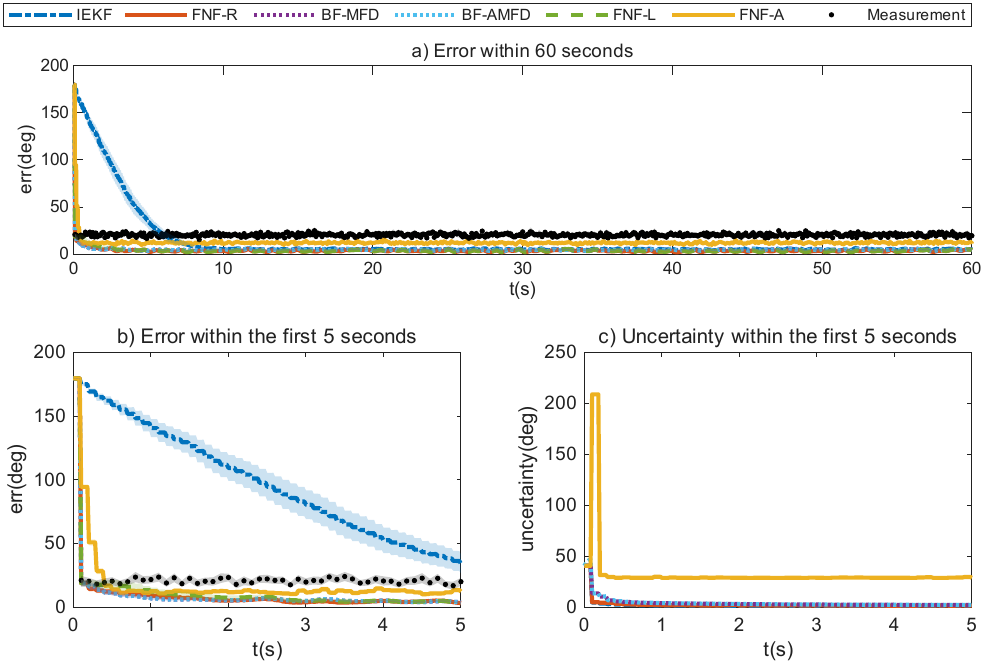}
		\caption{Average error for the case of large initial errors with non-isotropic measurement errors. The shadow represents an envelope of 95\% confidence. The attitude uncertainty is calculated as the square root of the first diagonal term of the attitude covariance matrix in the inertial frame as \cite{wang2020matrix}.}\label{fig:large_error}
	\end{center}
\end{figure}

\begin{table*}[htbp]
\centering
\caption{AE (±SD) and time consumption: large initial errors}

\label{tab:2}
\begin{tabular}{c|c|cccccc}
\hline
Measurement Model     &             & FNF-R& FNF-L &FNF-A& BF-MFD&BF-AMFD& IEKF         \\ \hline
{(i)}& est. err(°) & 6.0477±0.4795& 6.0512±0.4818&20.7879±0.6656& 5.9893±0.4353&5.9660±0.4236& 18.5777±15.4668\\
                      & CPU time(s) & 0.5646±0.1670& 0.8381±0.0889&-& 15.7117±1.0018&5.0603±0.3695& 0.3229±0.0373\\ \hline
{(ii)}& est. err(°) & 4.1318±0.1611& 4.132±0.1605&8.2055±0.1972& 4.1279±0.1597&4.1228±0.1597& 17.1230±8.4887\\
                      & CPU time(s) & 0.5429±0.1600& 0.8224±0.0611&-& 62.4333±4.9024&4.7197±0.3973& 0.4520±0.0412\\ \hline
{(iii)}& est. err(°) & 4.1889±0.2782& 4.4379±0.4843&12.4253±0.8342& 5.2067±0.4618&5.1971±0.4524& 12.8828±3.2472\\
                      & CPU time(s) & 0.4853±0.0486& 0.7852±0.0513&-& 15.0536±1.2947&4.9301±0.4447& 0.3118±0.0499\\ \hline
\end{tabular}
\end{table*}

As shown in Fig. \ref{fig:large_error}, the two proposed filters achieve, in cases i) and ii), almost identical estimation accuracy and convergence rate with BF-MFD and BF-AMFD. In case iii), their discrepancy in accuracy is mainly caused by different approximation methods for likelihood distributions. All MFD-based filters perform better than IEKF in all three cases. More importantly, both proposed filters consume significantly less time than BF-MFD and BF-AMFD, which is shown in Fig. \ref{fig:time}. The two proposed filters consume slightly more time than IEKF due to the involvement of matrix inversion in calculating the prior distribution, and much less time than BF-MFD and BF-AMFD (about one to two orders of magnitude less than BF-MFD and about one order of magnitude less than BF-AMFD) due to skipping the step of solving complex nonlinear equations in the moment matching as needed by BF-MFD and BF-AMFD. When comparing the results of cases (i) and (ii), it can be seen that a larger initial and measurement uncertainty leads to a slightly greater loss of accuracy for the FNF-R than for the BF-MFD and BF-AMFD. More specifically, the accuracy loss increases from about $0.1\%$ for case (ii) to about $1\%$ for case (i), which is insignificant in practice.

In Section \ref{sec:filter_mec}, we prove that the uncertainty of MFD-based Bayesian filters (measured by $1/\kappa^+$) is correlated with $\Delta \bar{\theta}$, while the uncertainty of filters with CGDs is independent of the estimation error. Therefore, in Fig. \ref{fig:large_error}, the uncertainties of the four filters with MFDs decrease in a manner compatible with estimation errors, while the uncertainty of the IEKF decreases rapidly even if the estimation error is still large. Additionally, the better transient response of the FNF-R compared with the IEKF for large initial errors is supported by the stability analysis in Section \ref{sec:stability}, which indicates that the proposed filter has a larger attraction region.

Compared with the FNF-R and FNF-L, the estimation accuracy of the FNF-A displays a significant decline. The simulation of case iii) shows that the convergence rate is slightly slower than those of FNF-R and FNF-L, and, more importantly, the estimation errors of the FNF-A cannot converge to the same level as the FNF-R and FNF-L, as shown in Fig. \ref{fig:large_error}. The ablation experiments with the FNF-A demonstrate the importance of introducing the MFD to represent uncertainty and assigning an appropriate weight to the Wahba problem when dealing with noisy measurements.

\begin{figure}[htbp]
	\begin{center}
		\includegraphics[width=\linewidth]{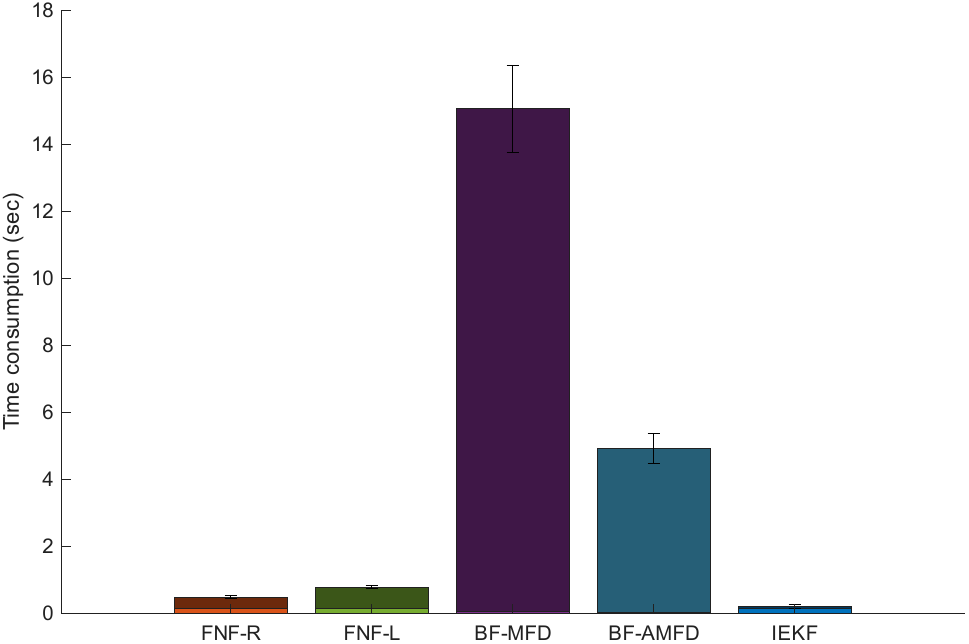}
		\caption{Time consumption: large initial error with non-isotropic measurement error. For each bar, the darker part on the top represents the time consumption of the prior for the proposed filters or propagation for BF-MFD and IEKF, and the lighter part on the bottom represents the time consumption of the likelihood and the posterior for the proposed filters or the measurement update for BF-MFD and IEKF.}\label{fig:time}
	\end{center}
\end{figure}





 \subsection{Direct attitude measurements with non-isotropic noise}

Finally, we test cases of direct attitude measurements with non-isotropic noise. The direct attitude measurement model for the four filters with MFDs is given by
\begin{equation}\label{eq:MF_dmea}
    R^m_k = R_k \delta R,\ \delta R \sim \mathcal{M}(N^m),
\end{equation}
where $N^m \in \mathbb{R}^{3\times 3}$ and $(N^m)^T = N^m$. The corresponding attitude measurement model for IEKF is given by
\begin{equation}\label{eq:ga_dmea}
     R^m_k = R_k \exp(\xi^{\wedge}),\ \xi \sim \mathcal{N}(0,P^m).
\end{equation}
All six filters will be tested with noisy measurements generated by (\ref{eq:MF_dmea}) and (\ref{eq:ga_dmea}), respectively. The measurement noises are set as: (i) $N^m = {\rm diag}[100, 0,0]$; (ii) $P^m = \text{diag} [10, 0.01, 0.01]$. The parameters $P^m$ and $N^m$ are chosen to be similar and correspond to the case that the rotation about the $b_1$ axis of $\mathcal{F_B}$ is completely unknown. The rest parameter settings are the same as Section \ref{sec:small_ie}. The results are summarized in Table \ref{tab:4} and Fig. \ref{fig:non_isotripic}.
\begin{figure}[htbp]
	\begin{center}
		\includegraphics[width=\linewidth]{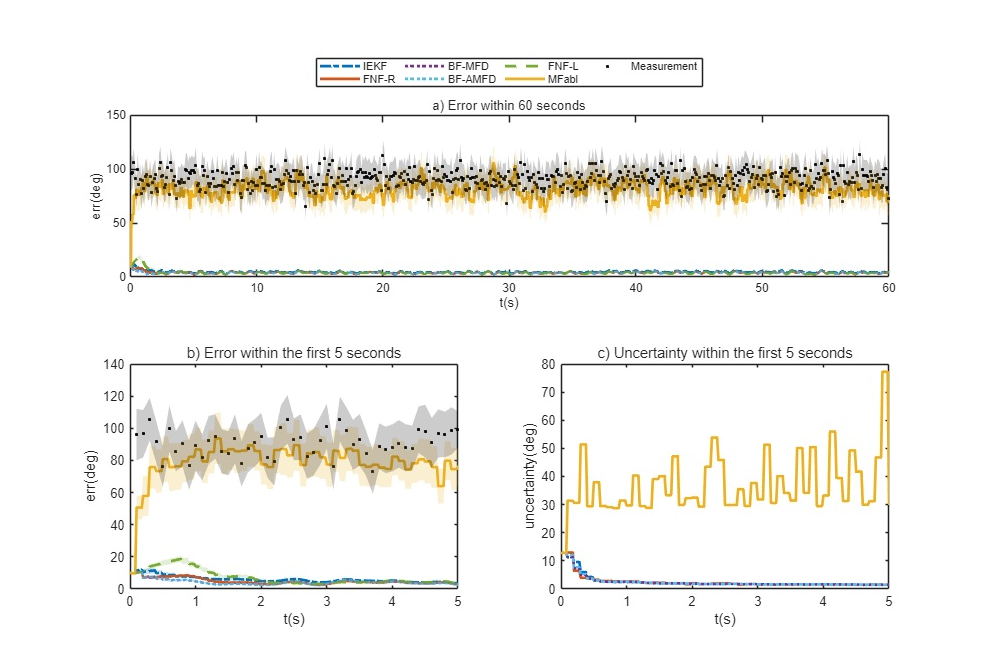}
		\caption{Average error for the case of direct attitude measurements with non-isotropic noise: $N^m = {\rm diag}[100, 0,0]$.}\label{fig:non_isotripic}
	\end{center}
\end{figure}

\begin{table*}[htbp]
\centering
\caption{AE (±SD) and time consumption: direct attitude measurements with non-isotropic noise}

\label{tab:4}
\begin{tabular}{c|c|cccccc}
\hline
Measurement Model     &             & FNF-R& FNF-L &FNF-A& BF-MFD&BF-AMFD& IEKF         \\ \hline
{(i)}& est. err(°) & 3.9009±0.2270& 4.0630±0.2800 &81.3657±3.9599& 3.8311±0.2090&3.8279±0.2074& 4.7197±0.3888\\
                      & CPU time(s) & 
0.3685±0.0685& 0.7335±0.1455&-& 60.6803±4.3024&4.6566±0.4562& 0.4423±0.1347\\ \hline
 {(ii)}& est. err(°) 
& 3.2629±0.1230& 3.3728±0.1293&82.6802±1.8590& 3.2524±0.1164&3.2480±0.1160& 3.7206±0.1964\\ & CPU time(s) & 0.4878±0.2684& 0.8722±0.3697&-& 60.7077±4.8772&4.6037±0.4732& 0.4641±0.1959\\ \hline
\end{tabular}
\end{table*}
The two proposed filters maintain almost the same estimation accuracy as, and have significantly faster computation speeds than BF-MFD and BF-AMFD. All four MFD-based filters achieve smaller steady-state errors compared to IEKF. This indicates that although FNF-R and FNF-L utilize the linearized error systems, they still perform much better than IEKF when dealing with large measurement uncertainty. The ablation experiments show that the accuracy reduction of the FNF-A is more obvious when dealing with very large measurement uncertainty, which highlights the importance of MFDs in such situations.
\section{Conclusion}
\label{sec: Conclusion}
This paper addressed the fast nonlinear filtering problem with MFDs on $SO(3)$. Two key properties associated with MFD-based attitude filters are revealed by analyzing the evolution of the distributions on $SO(3)$ along Bayes' rule: (i) the difference in mean attitudes of the prior and the likelihood distributions influence the uncertainty of 
the posterior distribution for rotations around a specific axis; (ii) MFD-based filters can filter out wrong information with low confidence faster than filters with CGDs. This finding is based on analytic expressions instead of a limited number of numerical examples or intuitive inferences. Therefore, it underpins the general validity and advantage of attitude filters with MFDs for estimation on $SO(3)$.

Two filters with MFDs are thereby proposed, and they retain the aforementioned properties but reduce the computational burden by a significant amount compared to previous MFD-based attitude filters. The proposed filters represent the attitude uncertainty by the left- and right-invariant errors characterized by MFDs respectively. Moreover, the filter with right-invariant error is proved to be almost globally asymptotically stable as an observer on $SO(3)$, and to be locally exponentially stable for the case of single-axis rotation, which demonstrates the relationship between the two properties and the convergence rate. The filter with right-invariant errors performs as well as the Bayesian filter with MFDs presented in \cite{leeBayesianAttitudeEstimation2018} in challenging simulation examples but consumes much less computation time (by up to about two orders of magnitude compared with previous MFD-based attitude filters). The filter with left-invariant errors is slightly less accurate than the previous Bayesian filter with MFDs. Therefore, the two proposed filters are efficient alternatives to the existing Bayesian filters with MFDs for attitude estimation requiring high computational efficiency.

\section*{Appendix A}\label{app:A}
\setcounter{equation}{0}
\renewcommand{\theequation}{A\arabic{equation}}
The proof is based on the mathematical induction. Assume that the posterior estimates at $t_{k-1}$ are $\hat{R}_{k-1} =\exp\left( {\hat{\theta}_{k-1} b_3^{\wedge}} \right),\ N_{k-1}={\rm diag}[\kappa_{k-1},\kappa_{k-1},\kappa^{*}_{k-1}]$. The assumption is clearly valid at $t_1=1$. For $k>1$, the prior estimates provided by Algorithm \ref{alg:right_est} are
\begin{align*}
    &\hat{R}_{k|k-1}\!=\!\hat{R}_{k-1}\!\exp\!{\left( \tilde{w}_kb_3^{\wedge} \right)}\!=\! \exp\left[\!  {\left(\hat{\theta}_{k-1}\! +\!\tilde{w}_k h\right) b_3^{\wedge}} \!\right],\\
    &P_{k|k-1}  = P_{k - 1}  + \hat R_{k - 1} {Q_k}{\left( {\hat R_{k - 1} } \right)^T}\\
    &={\rm diag}\!\left[ \frac{1}{{\kappa _{k - 1}+ \kappa _{k - 1}^{ *}}}, \frac{1}{{\kappa _{k - 1}+ \kappa _{k - 1}^{ *}}},\frac{1}{2\kappa_{k-1}}\! \right]\!+\!Q_k,\\
\end{align*}
and $N_{k|k-1}$ can then be derived by (\ref{eq:N2P}) as
\begin{align*}
    N_{k|k-1}  ={\rm diag}[\kappa_{k|k-1},\kappa_{k|k-1},\kappa^{*}_{k|k-1}],
\end{align*}
where
\begin{align*}
    &\kappa_{k|k-1}=\frac{{\kappa _{k - 1} }}{{1 + 2\kappa _{k - 1} \sigma _k}},\\
    &\kappa_{k|k-1}^*= \frac{{\kappa _{k - 1}  + \kappa _{k - 1}^{ *}}}{{1 + {\sigma_k^{*}}\left( {\kappa _{k - 1}  + \kappa _{k - 1}^{ *}} \right)}} - \frac{{\kappa _{k - 1} }}{{1 + 2\kappa _{k - 1} \sigma _k}}.
\end{align*}
For measurement update, $F_k$ is calculated as
\begin{align*}
 &F_k = N_k^m \hat{R}_k^m + N_{k|k-1} \hat{R}_{k|k-1}\\
 &= N_k^m \exp \left( \theta^m_k b_3 \right)+ N_{k|k-1} \exp \left( \hat{\theta} _{k|k - 1} b_3 \right)\\
 &=\left[ {\begin{array}{*{20}{c}}
{\kappa _k^m\cos \theta _k^m}&{ - \kappa _k^m\sin \theta _k^m}&{}\\
{\kappa _k^m\sin \theta _k^m}&{\kappa _k^m\cos \theta _k^m}&{}\\
{}&{}&{\kappa _k^{m*}}
\end{array}} \right] +\\
&\left[ {\begin{array}{*{20}{c}}
{{\kappa _{k|k - 1}}\cos {\hat{\theta} _{k|k - 1}}}&{ - {\kappa _{k|k - 1}}\sin {\hat{\theta} _{k|k - 1}}}&{}\\
{{\kappa _{k|k - 1}}\sin {\hat{\theta} _{k|k - 1}}}&{{\kappa _{k|k - 1}}\cos {\hat{\theta} _{k|k - 1}}}&{}\\
{}&{}&{\kappa _{k|k - 1}^*}
\end{array}} \right].
\end{align*}
Denote by $\Delta \theta_{k|k-1}=\theta^m_k-\theta_{k|k-1}$, and we have
\begin{align*}
 \begin{array}{l}
\kappa _k^m\cos \theta _k^m + {\kappa _{k|k - 1}}\cos {{\hat \theta }_{k|k - 1}} \\
= \kappa _k^m\cos \left( {{{\hat \theta }_{k|k - 1}} + \Delta {\theta _{k|k-1}}} \right) + {\kappa _{k|k - 1}}\cos {{\hat \theta }_{k|k - 1}}\\
 = \left( {{\kappa _{k|k - 1}} + \kappa _k^m\cos \Delta {\theta _{k|k-1}}} \right)\cos {{\hat \theta }_{k|k - 1}} - \kappa _k^m\sin \Delta {\theta _{k|k-1}}\sin {{\hat \theta }_{k|k - 1}}\\
 = {\kappa _k}\cos \left( {{{\hat \theta }_{k|k - 1}} + \Delta {{\hat \theta }_k}} \right)
\end{array},\\
\begin{array}{l}
\kappa _k^m\sin {\theta _k} + {\kappa _{k|k - 1}}\sin {{\hat \theta }_{k|k - 1}} \\
= \kappa _k^m\sin \left( {{{\hat \theta }_{k|k - 1}} + \Delta {\theta _{k|k-1}}} \right) + {\kappa _{k|k - 1}}\sin {{\hat \theta }_{k|k - 1}}\\
 = \left( {{\kappa _{k|k - 1}} + \kappa _k^m\cos \Delta {\theta _{k|k-1}}} \right)\sin {{\hat \theta }_{k|k - 1}} + \kappa _k^m\sin \Delta {\theta _{k|k-1}}\cos {{\hat \theta }_{k|k - 1}}\\
 = \kappa _k \sin \left( {{{\hat \theta }_{k|k - 1}} + \Delta \hat{\theta }_k } \right),
\end{array}
\end{align*}
where
\begin{align*}
    &\kappa_k = \sqrt {{{(\kappa _{k|k-1} )}^2} + {{(\kappa _k^m)}^2} + 2\kappa  _{k|k-1} \kappa _k^m\cos (\Delta \theta_k)} ,\\
    &\Delta \hat{\theta}_k = {\rm{sign}}(\Delta \theta_k)\arccos \left( {\frac{{\kappa  _{k|k-1}  + \kappa _k^m\cos (\Delta \theta_k )}}{{\kappa  _{k} }}} \right).
\end{align*}
Then, $F_k$ can be rewritten as
\begin{equation*}
    \begin{array}{l}
{F_k} \\
 = \left[ {\begin{array}{*{20}{c}}
{{\kappa _k}{I_2}}&{}\\
{}&{\kappa _k^{m*} + \kappa _{k|k - 1}^*}
\end{array}} \right]\exp \left[ {\left( {{{\hat \theta }_{k|k - 1}} + \Delta {{\hat \theta }_k}} \right){b_3}^ \wedge } \right]
\end{array},
\end{equation*}
equivalent to the proper left polar decomposition of $F_k$. Thus, the posterior estimates are given by
\begin{align*}
    &N_k = {\rm diag}[\kappa_k,\kappa _k^{m*} + \kappa _{k|k - 1}^* ],\\
    &\hat{R}_k = \exp \left[ {\left( {{{\hat \theta }_{k|k - 1}} + \Delta {{\hat \theta }_k}} \right){b_3}^ \wedge } \right],
\end{align*}
which, under the induction assumption, preserves the forms of the diagonal matrix and the uniaxial rotation, respectively. By the mathematical induction, Theorem \ref{th:uniaxial} is valid for all positive integers $k$.

\section*{Appendix B}\label{app:B}
\setcounter{equation}{0}
\renewcommand{\theequation}{B\arabic{equation}}
\subsubsection{The monotonicity of $V_k$}
We first prove the monotonicity of $V_k$. Since the angular velocity measurements are noise-free, we have $\Delta \theta_{k|k-1}=\theta_k-\hat{\theta}_{k|k-1} =\Delta \theta_{k-1}= \theta_{k-1} -\hat{\theta}_{k-1}$, and thus
\begin{equation*}
    \frac{V_{k|k-1}}{V_{k-1}} =\frac{\kappa_{k|k-1}}{\kappa_{k-1}}=\frac{1}{1+2\kappa_{k-1}\sigma_k}<1.
\end{equation*}
Defining an intermediate variable $\tau = \kappa^m_k/\kappa_{k|k-1}$, there is
\begin{align*}
&\frac{{{V_k}}}{{{V_{k|k - 1}}}} = \frac{{{\kappa _k}\left[ {1 - \cos \left( {{\theta _k} - {{\hat \theta }_k}} \right)} \right]}}{{{\kappa _{k|k-1} }\left[ {1 - \cos \left( {\Delta {\theta _{k|k - 1}}} \right)} \right]}} = \\
&\sqrt {1 + {\tau^2} + 2\tau\cos \left( {\Delta {\theta _{k|k-1}}} \right)} \frac{{1 - \cos \left( {\Delta {\theta _{k|k-1}} - \Delta {{\hat \theta }_k}} \right)}}{{1 - \cos \left( {\Delta {\theta _{k|k - 1}}} \right)}},
\end{align*}
where $\Delta \hat{\theta}_k=\hat{\theta}_{k|k-1}-\hat{\theta}_k$. Invoking the noise-free measurement $R_k^m = R_k$, we have $\theta^m_k -\hat{\theta}_{k|k-1} = \Delta \theta_{k|k-1}$, and $\cos(\Delta {{\hat \theta }_k})$ as well as $\sin(\Delta {{\hat \theta }_k})$ can be rewritten as
\begin{align*}
 &
\cos \left[ {{\rm{sign}}(\Delta {\theta _{k|k - 1}})\arccos \left( {\frac{{{\kappa _{k|k - 1}} + \kappa _k^m\cos (\Delta {\theta _{k|k - 1}})}}{{{\kappa _k}}}} \right)} \right]\\
& = \frac{{{\kappa _{k|k - 1}} + \kappa _k^m\cos (\Delta {\theta _k})}}{{{\kappa _k}}} = \frac{{1 + \tau\cos (\Delta {\theta _{k|k - 1}})}}{{\sqrt {1 + {\tau^2} + 2\tau\cos \left( {\Delta {\theta _{k|k - 1}}} \right)} }},
\\
&\sin \left[ {{\rm{sign}}(\Delta {\theta _{k|k - 1}})\arccos \left( {\frac{{{\kappa _{k|k - 1}} + \kappa _k^m\cos (\Delta {\theta _{k|k - 1}})}}{{{\kappa _k}}}} \right)} \right]\\
&=\frac{{\tau{{\sin }}\Delta {\theta _{k|k-1}}}}{{\sqrt {1 + {\tau^2} + 2\tau\cos \left( {\Delta {\theta _{k|k-1}}} \right)} }}.
\end{align*}
The above equalities can be used to derive
\begin{align}\label{app_B:eq_1}
  & \cos \!\left( {\Delta {\theta _{k|k - 1}} \! - \! \Delta {{\hat \theta }_k}} \! \right) \! = \! \cos \! \Delta {\theta _{k|k - 1}} \!\cos \! \Delta {\hat \theta _k} \! + \! \sin  \!\Delta {\theta _{k|k - 1}} \!\sin \! \Delta {\hat \theta _k}\nonumber\\
 & =\frac{{\cos \Delta {\theta _{k|k - 1}} + \tau}}{{\sqrt {1 + {\tau^2} + 2\tau\cos \left( {\Delta {\theta _{k|k - 1}}} \right)} }},
\end{align}
 and thus
 \begin{align*}
     &\frac{{{V_k}}}{{{V_{k|k - 1}}}} = \frac{{\sqrt {1 + {\tau^2} + 2\tau\cos \left( {\Delta {\theta _{k|k - 1}}} \right)}  - \tau - \cos \Delta {\theta _{k|k - 1}}}}{{1 - \cos \left( {\Delta {\theta _{k|k - 1}}} \right)}}\\
     &\leq \frac{{1 + \tau -\tau - \cos \Delta {\theta _{k|k-1}}}}{{1 - \cos \Delta {\theta _{k|k-1}}}} = 1,
 \end{align*}
 where the equality holds if and only if $\Delta\theta_{k|k-1} = 0$, i.e., the error converges to zero. Therefore, we have
 \begin{equation}\label{app_B:eq_2}
     V_k \leq V_{k|k-1}=\frac{V_{k-1}}{1+2\kappa_{k-1}\sigma_k}<V_{k-1}.
 \end{equation}

\subsubsection{ Verification of local exponential stability}
 Next, we examine the boundedness of $V_k$. Using (\ref{app_B:eq_1}), $\cos(\theta_k-\hat{\theta}_{k})/\cos(\theta_k-\hat{\theta}_{k|k-1})$ can be rewritten as
\begin{equation}\label{app_B:eq_3}
    \frac{{1 + \tau/\cos(\Delta\theta_{k|k-1})}}{{{\sqrt {1 + {\tau^2} + 2\tau\cos \left( {\Delta {\theta _{k|k - 1}}} \right)} }}}\geq 1,
\end{equation}
where the equality holds if and only if $\Delta\theta_{k|k-1} = 0$. Equation (\ref{app_B:eq_3}) implies that the estimation error monotonically decreases as $k$ increases, i.e.,
\begin{equation}\label{app_B:eq_4}
   1-\cos(\theta_k-\hat{\theta}_{k}) < 1-\cos(\theta_k-\hat{\theta}_{k|k-1})=1-\cos(\theta_{k-1}-\hat{\theta}_{k-1}).
\end{equation}
Since the initial attitude satisfies ${\rm tr}(I_3-R_0\hat{R}_0^T)=2[1-\cos(\Delta\theta_0)]<2$, we have $\cos(\Delta\theta_0)\geq0$. By (\ref{app_B:eq_4}), the estimation error decreases monotonically, so $\cos(\Delta\theta_{k|k-1})\geq\cos(\Delta\theta_0)\geq0$ for all $k\geq1$. Therefore, from
\begin{equation*}
    \kappa_k=\sqrt{(\kappa_{k|k-1})^2+(\kappa_k^m)^2+2\kappa_{k|k-1}\kappa_k^m\cos(\Delta\theta_k)}
\end{equation*}
and the hypothesis $\kappa_k^m>\beta_1$, we have
\begin{equation}\label{app_B:eq_kappa_lb}
    \kappa_k\geq\kappa_k^m>\beta_1,\qquad k\geq1.
\end{equation}
Finally, we prove that the estimation error decreases exponentially. Using (\ref{app_B:eq_2}), (\ref{app_B:eq_kappa_lb}) and $\sigma_k>\alpha_1$, we have, for $k\geq2$,
\begin{equation*}
    V_k\leq\frac{V_{k-1}}{1+2\kappa_{k-1}\sigma_k}\leq\lambda V_{k-1},\qquad \lambda=\left(1+2\alpha_1\beta_1\right)^{-1}<1,
\end{equation*}
while for $k=1$ we have $V_1\leq V_0$ since $\kappa_0>0$. Hence
\begin{equation}\label{app_B:eq_5}
    V_k\leq V_1\lambda^{k-1}\leq\frac{V_0}{\lambda}\lambda^k,\qquad k\geq1.
\end{equation}
Noting $V_k=\frac{\kappa_k}{4}\|R_k-\hat R_k\|_F^2$ by (\ref{eq:de_ly}) and $V_0=\kappa_0[1-\cos(\Delta\theta_0)]$, we obtain, for $k\geq1$,
\begin{equation*}
    \|R_k-\hat R_k\|_F^2=\frac{4}{\kappa_k}V_k\leq\frac{4V_0}{\beta_1\lambda}\lambda^k=\frac{\kappa_0}{\beta_1\lambda}\lambda^k\|R_0-\hat R_0\|_F^2,
\end{equation*}
which verifies the theorem.
\section*{References}
\bibliographystyle{ieeetr}        
\bibliography{IEEEabrv,citedpaper} 

\begin{thebibliography}{10}

\bibitem{mahony2008nonlinear}
R.~Mahony, T.~Hamel, and J.-M. Pflimlin, ``Nonlinear complementary filters on
  the special orthogonal group,'' {\em IEEE Transactions on automatic control},
  vol.~53, no.~5, pp.~1203--1218, 2008.

\bibitem{zlotnik2018higher}
D.~E. Zlotnik and J.~R. Forbes, ``Higher order nonlinear complementary
  filtering on lie groups,'' {\em IEEE Transactions on Automatic Control},
  vol.~64, no.~5, pp.~1772--1783, 2018.

\bibitem{wu2015globally}
T.-H. Wu, E.~Kaufman, and T.~Lee, ``Globally asymptotically stable attitude
  observer on so (3),'' in {\em 2015 54th IEEE Conference on Decision and
  Control (CDC)}, pp.~2164--2168, IEEE, 2015.

\bibitem{berkane2017hybrid}
S.~Berkane, A.~Abdessameud, and A.~Tayebi, ``Hybrid global exponential
  stabilization on so (3),'' {\em Automatica}, vol.~81, pp.~279--285, 2017.

\bibitem{berkane2017hybridbias}
S.~Berkane, A.~Abdessameud, and A.~Tayebi, ``Hybrid attitude and gyro-bias
  observer design on $ so (3) $,'' {\em IEEE Transactions on Automatic
  Control}, vol.~62, no.~11, pp.~6044--6050, 2017.

\bibitem{wang2020hybrid}
M.~Wang and A.~Tayebi, ``Hybrid nonlinear observers for inertial navigation
  using landmark measurements,'' {\em IEEE Transactions on Automatic Control},
  vol.~65, no.~12, pp.~5173--5188, 2020.

\bibitem{wang2023nonlinear}
M.~Wang and A.~Tayebi, ``Nonlinear attitude estimation using intermittent and
  multirate vector measurements,'' {\em IEEE Transactions on Automatic
  Control}, vol.~69, no.~8, pp.~5231--5245, 2023.

\bibitem{leffertsKalmanFilteringSpacecraft1982}
E.~Lefferts, F.~Markley, and M.~Shuster, ``Kalman {{Filtering}} for
  {{Spacecraft Attitude Estimation}},'' {\em Journal of Guidance, Control, and
  Dynamics}, vol.~5, pp.~417--429, Sept. 1982.

\bibitem{markley2003attitude}
F.~L. Markley, ``Attitude error representations for kalman filtering,'' {\em
  Journal of guidance, control, and dynamics}, vol.~26, no.~2, pp.~311--317,
  2003.

\bibitem{barrauIntrinsicFilteringLie2015}
A.~Barrau and S.~Bonnabel, ``Intrinsic {{Filtering}} on {{Lie Groups With
  Applications}} to {{Attitude Estimation}},'' {\em IEEE Transactions on
  Automatic Control}, vol.~60, pp.~436--449, Feb. 2015.

\bibitem{barrauInvariantExtendedKalman2017}
A.~Barrau and S.~Bonnabel, ``The {{Invariant Extended Kalman Filter}} as a
  {{Stable Observer}},'' {\em IEEE Transactions on Automatic Control}, vol.~62,
  pp.~1797--1812, Apr. 2017.

\bibitem{barrauStochasticObserversLie2018}
A.~Barrau and S.~Bonnabel, ``Stochastic observers on {{Lie}} groups: A
  tutorial,'' in {\em 2018 {{IEEE Conference}} on {{Decision}} and {{Control}}
  ({{CDC}})}, pp.~1264--1269, IEEE, 2018.

\bibitem{barrau2016invariant}
A.~Barrau and S.~Bonnabel, ``The invariant extended kalman filter as a stable
  observer,'' {\em IEEE Transactions on Automatic Control}, vol.~62, no.~4,
  pp.~1797--1812, 2016.

\bibitem{barrauExtendedKalmanFiltering2020}
A.~Barrau and S.~Bonnabel, ``Extended {{Kalman Filtering With Nonlinear
  Equality Constraints}}: {{A Geometric Approach}},'' {\em IEEE Transactions on
  Automatic Control}, vol.~65, pp.~2325--2338, June 2020.

\bibitem{barrau2022geometry}
A.~Barrau and S.~Bonnabel, ``The geometry of navigation problems,'' {\em IEEE
  Transactions on Automatic Control}, vol.~68, no.~2, pp.~689--704, 2022.

\bibitem{guiQuaternionInvariantExtended2018}
H.~Gui and A.~H.~J. {de Ruiter}, ``Quaternion {{Invariant Extended Kalman
  Filtering}} for {{Spacecraft Attitude Estimation}},'' {\em Journal of
  Guidance, Control, and Dynamics}, vol.~41, pp.~863--878, Apr. 2018.

\bibitem{barrau2015ekf}
A.~Barrau and S.~Bonnabel, ``An {EKF-SLAM} algorithm with consistency
  properties,'' {\em arXiv preprint arXiv:1510.06263}, 2015.

\bibitem{liuInGVIOConsistentInvariant2023}
C.~Liu, C.~Jiang, and H.~Wang, ``{{InGVIO}}: {{A Consistent Invariant Filter}}
  for {{Fast}} and {{High-Accuracy GNSS-Visual-Inertial Odometry}},'' {\em IEEE
  Robotics and Automation Letters}, vol.~8, pp.~1850--1857, Mar. 2023.

\bibitem{songRightInvariantExtended2022}
Y.~Song, Z.~Zhang, J.~Wu, Y.~Wang, L.~Zhao, and S.~Huang, ``A {{Right Invariant
  Extended Kalman Filter}} for {{Object Based SLAM}},'' {\em IEEE Robotics and
  Automation Letters}, vol.~7, pp.~1316--1323, Apr. 2022.

\bibitem{chirikjian2011stochastic}
G.~S. Chirikjian, {\em Stochastic models, information theory, and Lie groups,
  volume 2: Analytic methods and modern applications}, vol.~2.
\newblock Springer Science \& Business Media, 2011.

\bibitem{barrauInvariantKalmanFiltering2018}
A.~Barrau and S.~Bonnabel, ``Invariant {{Kalman Filtering}},'' {\em Annual
  Review of Control, Robotics, and Autonomous Systems}, vol.~1, pp.~237--257,
  May 2018.

\bibitem{Wang2006ErrorPO}
Y.~Wang and G.~S. Chirikjian, ``Error propagation on the {{Euclidean}} group
  with applications to manipulator kinematics,'' {\em IEEE Transactions on
  Robotics}, vol.~22, no.~4, pp.~591--602, 2006.

\bibitem{wang2008nonparametric}
Y.~Wang and G.~S. Chirikjian, ``Nonparametric second-order theory of error
  propagation on motion groups,'' {\em The International journal of robotics
  research}, vol.~27, no.~11-12, pp.~1258--1273, 2008.

\bibitem{wolfe2011bayesian}
K.~C. Wolfe, M.~Mashner, and G.~S. Chirikjian, ``Bayesian fusion on lie
  groups,'' {\em Journal of Algebraic Statistics}, vol.~2, no.~1, 2011.

\bibitem{Bourmaud2014ContinuousDiscreteEK}
G.~Bourmaud, R.~M{\'e}gret, M.~Arnaudon, and A.~Giremus, ``Continuous-discrete
  extended kalman filter on matrix lie groups using concentrated gaussian
  distributions,'' {\em Journal of Mathematical Imaging and Vision}, vol.~51,
  pp.~209 -- 228, 2014.

\bibitem{wang2021matrix}
W.~Wang and T.~Lee, ``Matrix {{Fisher}}--{{Gaussian}} distribution on
  $\mathrm{SO}(3)\times\mathbb{R}^n$ and {{Bayesian}} attitude estimation,''
  {\em IEEE Transactions on Automatic Control}, vol.~67, no.~5, pp.~2175--2191,
  2021.

\bibitem{mardiaDirectionalStatistics1999}
K.~V. Mardia and P.~E. Jupp, eds., {\em Directional {{Statistics}}}.
\newblock Wiley {{Series}} in {{Probability}} and {{Statistics}}, Hoboken, NJ,
  USA: John Wiley \& Sons, Inc., Jan. 1999.

\bibitem{leyModernDirectionalStatistics2017}
C.~Ley and T.~Verdebout, {\em Modern {{Directional Statistics}}}.
\newblock {Chapman and Hall/CRC}, 1~ed., Aug. 2017.

\bibitem{leyAppliedDirectionalStatistics2019}
C.~Ley and T.~Verdebout, eds., {\em Applied Directional Statistics}.
\newblock Boca Raton, FL: CRC Press, 2019.

\bibitem{downsOrientationStatistics1972}
T.~D. Downs, ``Orientation statistics,'' {\em Biometrika}, vol.~59, no.~3,
  pp.~665--676, 1972.

\bibitem{khatriMisesFisherMatrixDistribution1977}
C.~G. Khatri and K.~V. Mardia, ``The von {{Mises-Fisher Matrix Distribution}}
  in {{Orientation Statistics}},'' {\em Journal of the Royal Statistical
  Society: Series B (Methodological)}, vol.~39, pp.~95--106, Sept. 1977.

\bibitem{binghamAntipodallySymmetricDistribution1974}
C.~Bingham, ``An antipodally symmetric distribution on the sphere,'' {\em The
  Annals of Statistics}, pp.~1201--1225, 1974.

\bibitem{leeBayesianAttitudeEstimation2018}
T.~Lee, ``Bayesian attitude estimation with the matrix {{Fisher}} distribution
  on {{SO}} (3),'' {\em IEEE Transactions on Automatic Control}, vol.~63,
  no.~10, pp.~3377--3392, 2018.

\bibitem{kurzUnscentedMisesFisher2016}
G.~Kurz, I.~Gilitschenski, and U.~D. Hanebeck, ``Unscented von
  {{Mises}}--{{Fisher Filtering}},'' {\em IEEE Signal Processing Letters},
  vol.~23, pp.~463--467, Apr. 2016.

\bibitem{wangBinghamGaussianDistributionBackslash2021}
W.~Wang and T.~Lee, ``Bingham-{{Gaussian Distribution}} on $\mathbb
  {S}^{3}\times\mathbb {R}^{n} $ for {{Unscented Attitude Estimation}},'' in
  {\em 2021 {{IEEE International Conference}} on {{Multisensor Fusion}} and
  {{Integration}} for {{Intelligent Systems}} ({{MFI}})}, pp.~1--7, IEEE, 2021.

\bibitem{wang2020matrix}
W.~Wang and T.~Lee, ``Matrix {{Fisher}}-{{Gaussian}} distribution on $\mathrm
  {SO}(3)\times\mathbb {R}^ n $ for attitude estimation with a gyro bias,''
  {\em arXiv preprint arXiv:2003.02180}, 2020.

\bibitem{leeBayesianAttitudeEstimation2018b}
T.~Lee, ``Bayesian {{Attitude Estimation}} with {{Approximate Matrix Fisher
  Distributions}} on {{SO}}(3),'' in {\em 2018 {{IEEE Conference}} on
  {{Decision}} and {{Control}} ({{CDC}})}, (Miami Beach, FL), pp.~5319--5325,
  IEEE, Dec. 2018.

\bibitem{markleyAttitudeDeterminationUsing1988}
F.~L. Markley, ``Attitude determination using vector observations and the
  singular value decomposition,'' {\em Journal of the Astronautical Sciences},
  vol.~36, no.~3, pp.~245--258, 1988.

\bibitem{chirikjianEngineeringApplicationsNoncommutative2000}
G.~S. Chirikjian, {\em Engineering Applications of Noncommutative Harmonic
  Analysis: With Emphasis on Rotation and Motion Groups}.
\newblock CRC press, 2000.

\bibitem{fisherDispersionSphere1953}
R.~A. Fisher, ``Dispersion on a sphere,'' {\em Proceedings of the Royal Society
  of London. Series A. Mathematical and Physical Sciences}, vol.~217, no.~1130,
  pp.~295--305, 1953.

\bibitem{sanyal2006globaloptimalattitudeestimation}
A.~K. Sanyal, T.~Lee, M.~Leok, and N.~H. McClamroch, ``Global optimal attitude
  estimation using uncertainty ellipsoids,'' 2006.

\bibitem{Lee2017matrix}
T.~Lee, ``Matrix {Fisher} distribution.'' Available:
  \url{https://github.com/fdcl-gwu/Matrix-Fisher-Distribution}, 2017.
\newblock [Online].

\bibitem{Wahba1965ALS}
G.~Wahba, ``A least squares estimate of satellite attitude,'' {\em Siam
  Review}, vol.~7, pp.~409--409, 1965.

\bibitem{berkane2017hybridlong}
S.~Berkane, {\em Hybrid attitude control and estimation on SO (3)}.
\newblock PhD thesis, The University of Western Ontario (Canada), 2017.

\bibitem{loveGaussianStatisticsPalaeomagnetic2003}
J.~J. Love and C.~G. Constable, ``Gaussian statistics for palaeomagnetic
  vectors,'' {\em Geophysical Journal International}, vol.~152, pp.~515--565,
  Mar. 2003.

\end{thebibliography}

\end{document}